%% file: main.tex
\documentclass[11pt,letterpaper]{article}
\usepackage[utf8]{inputenc}
\usepackage{latexsym}
\usepackage{amsmath}
\usepackage{amssymb}
\usepackage{amsthm}
\usepackage{mathtools}
\usepackage{mathrsfs}
\usepackage{epsfig}
\usepackage{hyperref}
\usepackage{subfig}
\usepackage{xcolor}
\usepackage{cleveref}
\usepackage{physics}
\usepackage{bm}
\usepackage[numbers,sort&compress]{natbib}
\usepackage{enumerate}
\usepackage{comment}

\usepackage{sectsty}
\allsectionsfont{\boldmath}

\makeatletter
\renewcommand*\l@section[2]{%
  \ifnum \c@tocdepth >\z@
    \addpenalty\@secpenalty
    \addvspace{1.0em \@plus\p@}%
    \setlength\@tempdima{1.5em}%
    \begingroup
      \parindent \z@
      \rightskip \@pnumwidth
      \parfillskip -\@pnumwidth
      \leavevmode
      \bfseries\boldmath
      \advance\leftskip\@tempdima
      \hskip -\leftskip
      #1\nobreak\hfil
      \nobreak\hb@xt@\@pnumwidth{\hss #2\kern-\p@\kern\p@}\par
    \endgroup
  \fi
}
\makeatother

\usepackage[nottoc]{tocbibind}

\newtheorem{theorem}{Theorem}[section]
\newtheorem{corollary}[theorem]{Corollary}
\newtheorem{lemma}[theorem]{Lemma}

\newtheorem{definition}[theorem]{Definition}
\newtheorem{problem}[theorem]{Problem}

\newtheorem{fact}[theorem]{Fact}

\DeclareMathOperator*{\E}{\mathbb{E}}

\newcommand{\R}{\mathbb{R}}
\newcommand{\C}{\mathbb{C}}
\newcommand{\calS}{\mathcal{S}}
\newcommand{\calP}{\mathcal{P}}
\newcommand{\eps}{\epsilon}

\newcommand{\iprod}[1]{\left\langle #1 \right\rangle}

\renewcommand{\vec}[1]{{\bm{#1}}}

\renewcommand{\norm}[1]{\left\| #1 \right\|}
\newcommand{\opnorm}[1]{\norm{#1}_{\infty}}
\newcommand{\wh}{\widehat}

\newcommand{\tmax}{\tau_{\max}}
\newcommand{\tmin}{\tau_{\min}}

\newcommand{\tv}{\mathrm{tv}}
\newcommand{\mc}{\mathcal}
\newcommand{\nnz}{\mathrm{nnz}}

\newcommand{\Z}{\mathbb{Z}}

\newcommand{\dph}{D_{\mathrm{phase}}}

\newcommand{\jerry}[1]{\textcolor{blue}{[jerry: #1]}}
\newcommand{\ziyun}[1]{\textcolor{purple}{[ziyun: #1]}}
\newcommand{\joe}[1]{\textcolor{teal}{[Joe: #1]}}

\title{Lower Bounds for Learning Hamiltonians from Time Evolution}
\author{Ziyun Chen\thanks{University of Washington, ziyuncc@cs.washington.edu. Supported by a Simons investigator award 928589, and an NSF grant CCF-2203541.} \and Jerry Li\thanks{University of Washington, jerryzli@cs.washington.edu.} \and Joseph Slote\thanks{University of Washington, jslote@cs.washington.edu.}}

\begin{document}

\date{}

\maketitle
\thispagestyle{empty} 

\begin{abstract}
    Learning about a Hamiltonian $H$ from its time evolution $e^{-iHt}$ is a fundamental task in quantum science.
    A flurry of recent work has developed powerful new algorithms with provable guarantees for this task, for a variety of natural settings.
    Despite this, relatively little is known about lower bounds for learning Hamiltonians.
    In particular, in the natural setting where we assume $H$ is a $k$-local Hamiltonian on $n$ qubits, all existing algorithms require total evolution time at least $n^{\Omega (k)}$ to learn the parameters of $H$, and it remained open whether one could obtain even faster algorithms---or at the very least, if one could obtain better runtimes for simpler tasks, such as estimating a single designated coefficient of the Hamiltonian.

    In this work we show the answer is essentially \textit{no}, by obtaining strong lower bounds for these problems.
    We find that not only do $k$-local Hamiltonians require $n^{\Omega(k)}$ time evolution or interactions to learn, but also that in several senses, ``learning \textit{anything} about a Hamiltonian is just as hard as learning \textit{everything}.''
    In particular, we find the same $n^{\Omega(k)}$ lower bound holds for learning a \textit{single coefficient} of a $k$-local Hamiltonian $H$, even if the rest of $H$ is already known.
    We also show an $n^{\Omega(k)}$ lower bound for the task of \textit{effective} Hamiltonian learning, where one seeks only to learn a unitary that approximates the time evolution of $H$.
    Several related lower bounds, such as for general sparse (but not necessarily local) $H$ are also given.

    On the technical side, we make a new connection between Hamiltonian learning lower bounds and the analysis of Boolean functions, where we introduce a novel extremal property that may be of independent interest.
    Specifically, we demonstrate that Hamiltonian learning lower bounds follow from the existence of low-degree functions $f$ that simultaneously have (\textit{i}) bounded Fourier coefficients and (\textit{ii}) large absolute minima on the hypercube, \textit{i.e.}, large $\min_{x\in\{\pm 1\}^n}|f(x)|$.

\end{abstract}

\newpage
\thispagestyle{empty} 
\tableofcontents
\clearpage

\pagenumbering{arabic} 
\setcounter{page}{1}   

\input{intro}

\input{prelim}

\input{instance}

\input{subspace}

\bibliographystyle{abbrvnat}
\bibliography{bibliography}


\end{document}

%% file: intro.tex
\vspace{-2em}

\section{Introduction}
In this paper we consider the problem of \emph{learning Hamiltonians from time evolution}, a basic question in the study of quantum inference. 
Here we are given the ability to evolve states of our choice over time, according to an unknown quantum process as prescribed by Schr\"{o}dinger's equation, and the goal is to characterize properties of this process.
Such questions have a long history in quantum metrology and sensing~\cite{caves1981quantum,wineland1992spin,holland1993interferometric,bollinger1996optimal,valencia2004distant,lee2002quantum,mckenzie2002experimental,leibfried2004toward,de2005quantum}, quantum device engineering and testing~\cite{boulant2003robust,shulman2014suppressing,sheldon2016procedure,innocenti2020supervised,sundaresan2020reducing}, and many-body physics~\cite{wiebe2014quantum,wiebe2014hamiltonian,wang2017experimental,verdon2020quantum}. 

More formally, assume we have an $n$-qubit Hamiltonian
\[H = \sum_{P \in \{I, X, Y, Z \}^{\otimes n}} \alpha_P P\] encoding the interaction between qubits.
Here the sum is taken over all $n$-qubit Pauli matrices, and the Pauli coefficients $\alpha_P$ are the parameters that determine $H$.
The goal is to determine basic properties of these parameters, given the ability to evolve states according to $H$, \textit{i.e.}, (adaptively) apply the unitary $e^{-iHt}$, interleaved with general $H$-independent quantum channels and measurement (see Definition~\ref{def:general-access}).
There are two key and closely related measures of complexity for such an interaction: the total number of rounds of interaction, and the total evolution time required by the algorithm.
Clearly, for any algorithm in this setting to be considered efficient, we will want both quantities to be polynomial in the system size.

This question has recently received considerable attention from the quantum computing community.
A recent line of work has demonstrated a number of algorithms for learning Hamiltonians from time evolution~\cite{huang2023learning,bakshi2024structure,ma2024learning,hu2025ansatz}.
Typically, these papers focus on parameter recovery, \textit{i.e.}, their goal is to learn all of the coefficients to uniform error $\eps$, although there are also other natural learning objectives, such as learning good approximations to $H$ in spectral norm or other non-parametric notions of closeness, see \textit{e.g.}~\cite{bluhm2024hamiltonian}.
For parameter recovery, the best-known upper bounds require total time evolution cost and number of interactions with the Hamiltonian scaling as $\mathrm{poly} (\nnz(H)) / \eps$, where $\nnz(H)$ is an \emph{a priori} bound on the number of nonzero Pauli coefficients in the Hamiltonian.
This scaling with $\eps$ is commonly referred to as \emph{Heisenberg scaling}, and it is a folklore result that this is optimal, see \textit{e.g.}~\cite{bakshi2024structure}.

On the other hand, relatively little is known about the dependence on the locality $k$ of the Hamiltonian, or more generally, on the total number of nonzero terms $\nnz(H)$ in the Hamiltonian.
All known algorithms incur a polynomial cost in this parameter $\nnz(H)$. 
In particular, they require total evolution time $n^{O(k)}$ to learn an arbitrary $k$-local Hamiltonian.
However, prior to our work, there were essentially no lower bounds that scaled with $\nnz(H)$ for parameter learning.
While several works studied lower bounds for Hamiltonian learning from time evolution, the lower bounds they obtain either do not grow with system size, only hold under noisy measurements, or are for different notions of recovery, such as recovery in relative spectral norm (see Section~\ref{sec:related-work} for a more detailed discussion).
Indeed, obtaining a lower bound for parameter recovery that scales appropriately with $\nnz(H)$ was posed as an open question by~\cite{bakshi2024structure}, as well as at a recent workshop at FOCS~\cite{FOCS2024QuantumLearningOpenQuestions}.

\subsection{Our Results}

\paragraph{A lower bound for last parameter recovery.}
Our main result is a new lower bound for the parameter recovery problem that scales polynomially with the number of non-zero parameters of the Hamiltonian.
In fact, we show a strong lower bound against a significantly easier problem.
Namely, we show that it is hard to learn a single Pauli coefficient of the Hamiltonian, even if the algorithm knows which coefficient they are targeting ahead of time, and moreover, even if the algorithm is given \emph{the values of all of the other parameters of the Hamiltonian}.
We call this problem the \emph{last parameter recovery} problem (see~\Cref{def:last-param-recovery}).
Perhaps surprisingly, we show that this problem still requires runtime which scales polynomially with $\nnz(H)$.
More concretely:
\begin{theorem}[informal, see Corollaries~\ref{cor:last-param-general} and~\ref{cor:last-param-local}]
\label{thm:informal-last-parameter}
    Let $k \leq n$.
    Let $P$ be a known $1$-local Pauli matrix.
    Then, there exists a $k$-local Hamiltonian $M = \sum_{Q \neq P} \alpha_Q Q$ so that:
    \begin{itemize}
        \item  $|\alpha_Q| \leq 1$ for all $Q$, and
        \item For all $\beta\in[-1,1]$, any algorithm which, given access to $H=M+\beta P$, learns $\beta$ to additive error $\eta \leq 1/2$ with non-trivial probability, requires that
        \[\max(1,\eta\cdot T)\cdot m\geq \frac{(2n/k)^{\Omega(k)}}{\eta}\,,\]
        where $m$ is the total number of rounds of interaction with the Hamiltonian, and $T$ is the total evolution time required by the algorithm.
    \end{itemize}
\end{theorem}
We pause here to make several comments on our result.

\medskip
\noindent\textit{Interpreting the parameters.}
First, in \Cref{thm:informal-last-parameter} we stipulate that the Pauli coefficients are all uniformly bounded by $1$.
An \emph{a priori} upper bound on the scale of the coefficients is necessary to make this problem well-defined, and our choice of $1$ is canonical in the literature.
However, it turns out that this lower bound still holds if we significantly relax this bound for the unknown parameter $\beta$, enlarging $\eta$ proportionately.
While this would intuitively make the problem much easier---as a large ``spike'' on $P$ should make its presence more noticeable---we find that even when the size of the spike $\beta$ is $(2n / k)^{\Theta (k)}$, this problem remains hard.
We also observe that if we let $\eta \to 0$, then we recover the Heisenberg scaling bound, and moreover, we demonstrate that one cannot decouple the dependence on $1/\eta$ and $(2n / k)^{\Theta (k)}$.
See \Cref{cor:last-param-local} and \Cref{cor:last-param-general} for the formal statements. 

Our bound here can be interpreted as saying that either the total time evolution $T$ or the total number of interactions $m$ must scale as $\Omega (n / k)^{\Omega (k)}$, \textit{i.e.} polynomially with the number of total parameters in the system, even when the goal is to simply learn a single parameter.
In particular, our bound implies that as long as the minimum time resolution of the algorithm (\textit{i.e.}, the minimum amount of time that the algorithm evolves the Hamiltonian for in any step) is at least $(n / k)^{-o(k)}$, then the total evolution time of the algorithm must scale as $\Omega (n / k)^{\Omega (k)}$.
In most practical settings, it is unrealistic to assume anywhere near that degree of control, and in fact, usually anything below constant time resolution is considered impractical~\cite{bakshi2024learning}.
Thus, at the risk of slightly oversimplifying, our result essentially states that any realistic algorithm for learning even \emph{one} parameter, given all other possible knowledge about the Hamiltonian, still cannot hope to circumvent this polynomial dependence on $\nnz(H)$.

Next, note that our lower bound also immediately implies the same lower bound holds for learning all the parameters of the Hamiltonian, since our task is strictly easier.
Hence, our lower bound resolves the open question of~\cite{bakshi2024learning}, for which E. Tang owes us $\$ 2$.

\medskip
\noindent\textit{About the learning model.}
We emphasize that our lower bound applies to a very general class of Hamiltonian learning algorithms: we assume the learner accesses the Hamiltonian only through its time-evolution operator, but is otherwise allowed arbitrary classical and quantum processing. In particular, the learner may perform intermediate partial measurements and then use the resulting classical outcomes to choose future processing and Hamiltonian evolution times adaptively. Thus our model includes the standard setting of mid-circuit measurement and classical feedforward, which we formalize as a learning tree for Hamiltonian learning (see~\Cref{def:general-access}).

This is technically more general than the models considered in previous lower bounds such as~\cite{bluhm2024hamiltonian,huang2023learning,ma2024learning}. There, the learner proceeds through a sequence of \emph{experiments} in which Hamiltonian evolutions are interleaved with unitary controls, but the controls and evolution times within each experiment are fixed \textit{a priori}, and the resulting state is fully measured at the end of the experiment. Subsequent experiments may depend on previous measurement outcomes, but they do not allow quantum states to persist between adaptive operations. We refer to this more restrictive model as the \emph{sequential experiments model} (see~\Cref{def:access}).

The sequential experiments model is of course still natural and important, and it captures essentially all state-of-the-art algorithms. However, it does not capture all physically realizable learning strategies. Our lower bound continues to hold in the more general setting described above, including algorithms with intermediate partial measurement, classical feedforward, and adaptive choice of later Hamiltonian evolutions.

Relatedly, we also note that our lower bound holds irrespective of any amount of additional quantum memory.
It also holds even if the algorithm is allowed access to backwards time evolution of $H$, \textit{i.e.} the ability to run $e^{i H t}$ for $t > 0$.




\paragraph{Lower bounds for parameter recovery for sparse Hamiltonians.}
We now turn towards lower bounds for parameter estimation for sparse Hamiltonian ansatz.
Recently, several works have demonstrated that if we know that the support of the Hamiltonian $H$ lies within some small set; that is, $H = \sum_{P \in \mathcal{P}} \alpha_P P$ for some small set of Pauli matrices $\mathcal{P}$, then one can hope to learn $H$ to good accuracy in time which scales polynomially with $|\mathcal{P}|$ (in fact, linearly).
A natural question is whether or not this polynomial dependence on $|\mc P|$ is necessary.
Prior work of~\cite{ma2024learning} demonstrated that this was necessary for the sequential experiments model in the presence of SPAM noise, but left open the possibility that this dependence could be improved without this additional restriction.

Our next result shows that this dependence cannot be removed, even without SPAM noise, at least for the sequential experiments model:
\begin{theorem}
\label{thm:all-parameters-main}
    Let $\mathcal{P}$ denote a set of $n$-qubit Pauli matrices.
    Suppose there is an algorithm in the sequential experiments model (see~\Cref{def:access}) which, given $\mathcal{P}$, and time evolution to an unknown Hamiltonian $H = \sum_{P \in \mathcal{P}} \alpha_P P$ satisfying $|\alpha_P| \leq 1$ for all $P$, can output with probability $\geq 0.1$ an estimate $\{\wh \alpha_P \}_{P \in \mathcal{P}}$ so that $|\alpha_P - \wh \alpha_P| \leq 0.99$ for all $P \in \mathcal{P}$.
    Then, this algorithm must use at least $ m \geq \Omega \left( |\mathcal{P}| / n \right)$ rounds of interaction with $H$.
\end{theorem}
\noindent
We pause here to make several remarks about this result.
Our lower bound is for the total number of rounds of interaction $m$; said another way, we show that any circuit solving this problem requires high depth, which is arguably the strongest form of lower bound for this problem.
Our result implies that as long $|\mathcal{P}| = \Omega ( n^{1 + c})$ for some constant $c$, then the total number of interactions must grow polynomially in $|\mathcal{P}|$, the number of nonzero terms.
As an immediate corollary, our result implies that, in this regime, any algorithm that has constant minimum time resolution requires $\mathrm{poly} (|\mc P|)$ total evolution time.
In particular, for the class of $k$-local Hamiltonians for $k \geq 2$, our result implies that $\widetilde{O}((2n / k)^{k-1})$ interactions are necessary for this problem.

The closest comparison to this lower bound is the aforementioned lower bound of~\cite{ma2024learning}.
In~\cite{ma2024learning}, they show that any algorithm that learns the parameters to $\ell_1$ error $\eps$ in the presence of SPAM noise must use $\Omega(|\mathcal{P}|)$ total evolution time.
Our result is weaker quantitatively by a factor of $n$, however, ours is for a weaker notion of recovery, namely, recovery in $\ell_\infty$.
Additionally, as already mentioned above, our lower bound does not require the presence of SPAM noise.

Second, we note that this lower bound only holds against algorithms in the sequential experiments model.
We believe this restriction can be lifted with more sophisticated analysis, and we conjecture this lower bound can be extended to general algorithms.
However, we emphasize that this class of algorithms is already very general.
Indeed, all prior lower bound papers were also against this class of algorithm, and to our knowledge, all state-of-the-art algorithmic techniques for Hamiltonian learning fall within the purview of this lower bound.

\paragraph{Lower bounds for Hamiltonian approximation.}
In light of these lower bounds for any sort of meaningful parameter recovery, a natural question is whether or not we can hope to circumvent these bounds for other, perhaps weaker, notions of learning.
In other words, can we hope to recover something which is functionally equivalent to the Hamiltonian, but which may not be close in parameter distance?
If this sort of ``non-parametric'' recovery was possible, it would suffice for many downstream applications, as it would allow us to recover an object which is more or less functionally equivalent to the original Hamiltonian.

Unfortunately, we show that for the class of $k$-local Hamiltonians, this problem is still hard for algorithms in the sequential experiments model, even for a very weak notion of approximation:
\begin{theorem}[informal, see~\Cref{thm:effective-main}]
\label{thm:effective-informal}
    Let $t > 0$.
    Any algorithm in the sequential experiments model which, given time evolution to an unknown $k$-local Hamiltonian $H$ with spectral norm at most $1$, can output a non-trivial approximation to the quantum channel induced by $e^{-iHt}$, requires
    \begin{equation}
    \label{eq:effective-bound}
        m = \left( \frac{n}{k} \right)^{\Omega (k)}
    \end{equation}
    queries to the Hamiltonian.
\end{theorem}
We make several remarks about this bound.
First, this is a very weak notion of non-parametric approximation to $H$, as it does not require the algorithm to even recover $H$ directly, but only its induced dynamics.
In particular, when $t = O(1)$, this is a strictly easier task than the problem of learning $H$ to good spectral norm error.
Therefore, our result also immediately implies that the same lower bound holds even when the goal is to learn $H$ to good spectral norm error.
Relatedly, notice that our only assumption is that $H$ has bounded spectral norm.
This is an even stronger assumption than we made before for parameter recovery, since this immediately implies that all of its parameters are at most $1$.
Yet even with this stronger assumption, we show that the problem is still hard.

Second, as with the previous bound, this bound is against algorithms in the sequential experiments model, although we conjecture that this result remains true for more general classes of algorithms.

Third, for this problem, the trivial algorithm which outputs the identity channel always achieves error $\min (1, t)$ (see~\Cref{sec:prelim} for the formal definition of the error metric).
Our formal lower bound is quite quantitatively strong: in fact, we show that improving upon this trivial bound by a sufficiently large constant factor already requires $m$ to satisfy~\Cref{eq:effective-bound}.

\paragraph{An ``all-or-nothing'' phenomenon for Hamiltonian learning.} 

Our results---especially our results for last parameter learning and effective Hamiltonian learning---demonstrate that there is a qualitative ``all-or-nothing'' behavior for Hamiltonian learning.
Somewhat oversimplifying, our lower bounds state that to learn any information about a Hamiltonian $H$, then in many, if not most, regimes of interest, $\mathrm{poly}(\nnz(H))$ time is required.
On the other hand, we know from previous work of~\cite{bakshi2024learning,ma2024learning} that given $\mathrm{poly}(\nnz(H))$ evolution time, we can learn all of the information about the Hamiltonian.
In other words, one cannot extract any information from the Hamiltonian, without having essentially enough resources to extract all of the information from the Hamiltonian.

\subsection{Our Techniques}
\label{sec:techniques}
We now describe the intuition for our lower bound constructions.

\paragraph{Lower bound construction for last parameter learning.}
We first describe the construction for~\Cref{thm:informal-last-parameter}.
Intuitively, for the last parameter to be hard to learn, it makes sense for it to ``interfere'' with all of the other terms in the Hamiltonian, and thus we should choose it to be maximally non-commuting with the rest of the Hamiltonian.
Motivated by this, we consider an instance where the last unknown Pauli $\Delta$ is a single, $1$-local $X$ Pauli, and the rest of the Hamiltonian $H$ is supported only on $Z$ Paulis.

To formalize this, write $H$ as $H = \sum_{P \in \{I, Z\}^n} \alpha_P P$, for $|\alpha_P| \leq 1$.
In this case, $H$ is a diagonal matrix, and its entries are unnormalized Fourier transforms of a Boolean function specified by the $\alpha_P$.
Notably, even though $|\alpha_P| \leq 1$ for all $P$, the diagonal entries of $H$ may be exponentially large.
Now, the effect of adding $\Delta$ to this Hamiltonian is that now the resulting Hamiltonian is block-diagonal, with blocks of size $2$ corresponding to vertices of the hypercube that differ only in one coordinate.
Therefore, $e^{-i(H + \Delta) t}$ also shares this block-diagonal structure, and we will show that each block is close to the corresponding block of $e^{-iHt}$; this will imply that the two evolutions act indistinguishably from one another.

The key point is that this will be true so long as the difference in the diagonal elements on each block dominates the size of the off-diagonal perturbation due to $\Delta$ (see~\Cref{lem:matrix-bound}), and it suffices to demonstrate a choice of $\alpha_P$ that has this property.
We show that the existence of such an $\alpha_P$ is implied by the existence of a function on the hypercube with bounded Fourier coefficients, but which is exponentially large pointwise.
Motivated by this, we make the following definition:
\begin{definition}[Fourier-normalized absolute minima]
\label{def:fnam}
    Let $f: \{\pm 1\}^n \to \R$ be a multilinear function on the hypercube.
    We define the \emph{Fourier-normalized absolute minima} of $f$, denoted $\gamma (f)$, to be
    \[
    \gamma (f) = \frac{\min_{x \in \{\pm 1\}^n} |f(x)|}{\max_{S\subseteq[n]} |\widehat{f} (S)|} \; .
    \]
\end{definition}
\noindent
Here, $\widehat{f} (S)$ denote the standard $S$'th Boolean Fourier coefficient of $f$ (see~\Cref{sec:prelim} for formal definitions).
Formally, our first key insight (see~\Cref{thm:reduction}) is to show that the hardness for last parameter learning for $k$-local Hamiltonians can be reduced to demonstrating a degree $d = (k-1)$ polynomial with large Fourier-normalized absolute minima.

It then remains to construct $f$ with large Fourier-normalized absolute minima.
We give three constructions of such functions, which are suitable for different parameter regimes.
First, when $d = \Theta (n)$, we give a probabilistic construction of a function $f$ with $\gamma (f) = 2^{(1/2 - o(1)) n}$ (see Lemma~\ref{lem:boolean-instance}).
We then give a construction of a quadratic function with $\gamma (f) = \Omega (\sqrt{n})$, and show how to boost this to a construction of a degree $d$ function with $\gamma (f) = \Omega(\sqrt{n / d})^{\lfloor d / 2 \rfloor}$.
However, for $k > 2$, we believe the bound obtained from boosting is loose, and thus would yield loose quantitative bounds for the last parameter learning problem.

We are able to get a stronger lower bound in the regime of $d \leq O(n^{1/3})$.
In particular, we show (see~\Cref{thm:high-disc-poly}) that for such $d$, there is a degree-$d$ function $f$ satisfying
    \[
    \gamma (f) = \Omega \left(  \frac{1}{d^{5}} \cdot \left( \frac{n}{d} \right)^{\tfrac{d - 1}{2}} \right) \; .
    \]
Explicitly, this function $f$ is an elementary symmetric (multilinear) polynomial of degree $d$, for a number of variables $n^{*}$ close to $n$.
While the description of this function is quite simple, the analysis of its Fourier-normalized absolute minimum turns out to be quite non-trivial, and requires a careful analysis using specialized tools from the theory of orthogonal polynomials.

Our analysis of $\gamma(f)$ begins with the observation that because $f$ is symmetric, we may consider it as a univariate degree-$d$ polynomial, and in particular as a (binary) Krawtchouk polynomial $K_d(|x|;n)=f(x)$, where $|x|=|\{j:x_j=-1\}|$.
The main idea of our argument is to find an $n^*$ not too far from $n$ for which all the roots of $K_d(\,\cdot\,;n^*):\R\to\R$ are $\Omega_d(1)$-far from $\Z$.
Using that $K_d(\,\cdot\,;n^*)$ is degree-$d$ with well-spaced roots, it follows that $K_d(\,\cdot\,;n^*)$ is sufficiently large at every integer, and hence $\gamma(f)=\gamma(e^{(n^*)}_d)$ is large too.

To find a suitable $n^*$, we show the $j$\textsuperscript{th} root of $K_d(\,\cdot\,;n)$ is often far from $\Z$ as \textit{the number of variables} $n$ changes (but $d$ stays fixed), and then union bound over the $d$-many roots in $K_d(\,\cdot\,;n)$.
The key technical fact we must argue is a \textit{root increment lemma}: as $n$ changes, the $j$\textsuperscript{th} root shifts in a fixed direction by more than 1, but not too much more.
To do this we make use of the Krawtchouk polynomials' three-term recurrence and analyze the motion of eigenvalues of the associated Jacobi matrix.
This is the subject of \Cref{prop:root-inc}.

\paragraph{Lower bounds for learning sparse Hamiltonians and effective Hamiltonian learning.}
We now describe the main technique for our lower bound for learning sparse Hamiltonians, as well as our lower bound for learning effective Hamiltonians.
We show both via a shared framework.
The geometric insight is that the quantum state after $m$ queries to an $n$-qubit Hamiltonian can lie within a subspace of dimension at most $2^{m n}$, and thus by Holevo's inequality, we can extract at most $O(m n)$ bits of information about the Hamiltonian from a measurement of this quantum state.
Thus, by Fano's inequality, if it suffices to show that to solve the appropriate Hamiltonian estimation tasks, we need to extract many bits of classical information from the Hamiltonian.

This is fairly straightforward to do for parameter estimation for sparse Hamiltonians: here, if the set of Paulis with non-zero coefficients is $\mathcal{P}$, then we need to extract $\Omega (|\mc P| )$ bits if information to recover all of the Pauli coefficients.

Showing this for the effective Hamiltonian learning task is more non-trivial.
The key lemma we need is that there exists a large family of degree $k$ polynomials on the hypercube so that \textit{(i)} for all functions $f$ in this family, $|f(x)| \leq 1$ for all $x \in \{\pm 1\}^n$, and \textit{(ii)} for all $f \neq g$ in this family, we have that $|f(x) - g(x)| \geq \delta$ for some $x \in \{\pm 1\}^n$, for some constant $\delta > 0$.
This we do via the probabilistic method: we show that if we take our family of polynomials to have random Gaussian coefficients of the appropriate scale, then after throwing out all pairs of polynomials which violate Conditions (\textit{i}) and  (\textit{ii}), we still expect that the number of remaining polynomials should be large (see~\Cref{lem:packing-ell-infty}). 

We then use this lemma to directly exhibit a large family of $k$-local Hamiltonians so that their induced dynamics are all mutually far from each other, where the size of the family is exponential in $F(n, k)$, where
    \[
    F(n,k) \coloneqq \max_{t\in(0,\frac12 - \frac{1}{\log n})}  \min\left(n^{(1/2-t)k},\exp(cn^{2t})\right) \; .
    \]
By choosing a suitable $t$, and plugging this into the above framework, we obtain our overall lower bound.

\subsection{Related Work}
\label{sec:related-work}
There is a large literature on parameter learning Hamiltonians, both from thermal states~\cite{anshu2021sample,anshu2021efficient,haah2022optimal,haah2024learning,bakshi2024learning,fawzi2024certified,garcia2024estimation,narayanan2024improved,arunachalam2025testing}, as well as from time evolutions~\cite{ramsey1950molecular,lee2002quantum,giovannetti2004quantum,de2005quantum,shabani2011estimation,da2011practical,wiebe2014quantum,wiebe2014hamiltonian,bairey2019learning,zubida2021optimal,haah2022optimal,caro2024learning,odake2024higher,flynn2022quantum,gentile2021learning,haah2024learning,huang2023learning,li2024heisenberg,mirani2024learning,ni2024quantum,bakshi2024structure,ma2024learning,hu2025ansatz,arunachalam2025testing,zhao2025learning,sinha2025improved,francca2025learning}, the topic of this paper.
Much of the early literature on this topic was non-rigorous in nature.
The first rigorous guarantees were obtained in~\cite{haah2024learning}, and subsequently there has been a flurry of work on this topic.
Work of~\cite{huang2023learning} gave the first algorithms that achieved Heisenberg scaling.
Subsequent work of~\cite{li2024heisenberg,ni2024quantum,mirani2024learning} generalized these results to bosonic and fermionic Hamiltonians and work of~\cite{bakshi2024structure,ma2024learning} have also yielded efficient algorithms for the structure learning problem, and works of~\cite{zhao2025learning,dutkiewicz2024advantage} have achieved similar rates under different models of control or interaction.
However, almost all of the literature on this topic is for achieving upper bounds, all of which scale polynomially in the number of interacting terms in the Hamiltonian.

In terms of lower bounds, the literature is significantly sparser.
There are lower bounds for learning other classes of quantum transformations, such as Pauli channels~\cite{chen2024tight,chen2025efficient} and unitary channels~\cite{haah2023query}, as well as lower bounds for learning Hamiltonians from Gibbs states~\cite{anshu2021sample}, but these do not translate to lower bounds for the setting we consider here.
For the problem studied here, beyond the ``folklore'' lower bound of $\Omega (1 / \eps)$, the most relevant lower bounds for learning Hamiltonians from time evolution are given in~\cite{huang2023learning,ma2024learning,dutkiewicz2024advantage,bluhm2024hamiltonian,hu2025ansatz}.
However, the lower bounds in~\cite{dutkiewicz2024advantage,hu2025ansatz} does not scale with the number of interacting terms.
Rather,~\cite{dutkiewicz2024advantage} demonstrates that quantum control is necessary to obtain Heisenberg scaling, and this relationship is further refined in~\cite{hu2025ansatz}.
The lower bounds in~\cite{huang2023learning,ma2024learning} are for a setting where the every application of the Hamiltonian incurs some degree of SPAM noise; this effectively allows them to limit the number of useful rounds of interaction with the Hamiltonian.
In contrast, our bounds do not require any such noise.

Arguably the most relevant result to us is the lower bound of~\cite{bluhm2024hamiltonian}.
While they consider a different, locality testing problem, the proof of their Theorem 3.6 demonstrates that for a general $n$-qubit Hamiltonian $H$ with spectral norm at most $1$, $\Omega (2^{n / 2})$ rounds of interaction and $\Omega (2^{n / 2} / \eps)$ total time evolution is necessary for learning the target Hamiltonian to spectral norm error $\epsilon$, even with arbitrary measurements.
They also demonstrate stronger lower bounds for restricted classes of measurements.
However, their lower bound instance does not give meaningful lower bounds for the parameter estimation problem, as the Hamiltonians in their hard instance have Pauli coefficients which are exponentially small.
Additionally, when specialized for the $k$-local setting, their result only yields lower bounds which scale as $\Omega(2^{k / 2})$, and do not yield lower bounds of the form $n^{\Omega (k)}$.
In contrast, we demonstrate a lower bound which achieves the right qualitative scaling in both $n$ and $k$ simultaneously.
Notice that up to polynomial factors in $n$, and constants in the exponent, our lower bound for effective Hamiltonian learning almost directly generalizes their bound to the $k$-local setting, while matching their guarantees for general Hamiltonians.
This is because, as alluded to above, the effective Hamiltonian learning problem for time $t = \Theta (1)$ is strictly easier than the problem of learning $H$ to spectral norm error $\eps$.


\subsection{Outlook}

In this paper, we prove the first lower bounds for learning Hamiltonians from time evolution which scale polynomially with the number of non-zero terms in the Hamiltonian, both for general and for local Hamiltonians.
There are a number of natural next steps for this research direction.
One basic question is if we can obtain matching upper and lower bounds for learning arbitrary $k$-local Hamiltonians, or Hamiltonians with $M$ nonzero terms (we note that the upper bound of~\cite{ma2024learning} does not match our non-local lower bound, as their $\widetilde{O}$ hides terms which are exponential in $k$).
Another direction is if we can obtain instance-optimal bounds for the complexity of Hamiltonian learning for other sets of nonzero Pauli terms; for instance, when the Pauli terms are close to commuting, or for bosonic and/or fermionic Hamiltonians.
Some of the upper bounds, \textit{e.g.}~\cite{huang2023learning,bakshi2024structure} do naturally adapt to some extent to the structure of the interaction terms, but the rates they obtain are likely not tight in general, and we do not have lower bound instances for Hamiltonians with general sets of sparse interaction terms.
Another basic question is if one can establish similar lower bounds for learning from thermal states of Hamiltonians, matching the upper bounds of~\cite{bakshi2024learning}.
Finally, on the more technical side, a natural question is whether the bounds we established for Fourier-normalized absolute minima are tight for low-degree functions.

%% file: prelim.tex
\section{Preliminaries}
\label{sec:prelim}

\paragraph{Notation.}
For any matrix $M$, and any $p \in \R$, we let $\norm{M}_p$ denote the Schatten-$p$ norm of $M$.
For any subspace $V$, we let $\Pi_V$ denote the projection operator onto this subspace.
For any Pauli operator $P \in \{I, X, Y, Z\}^{\otimes n}$, we denote its coordinates via $P = \otimes_{i = 1}^n P_i$, and for any vector of coefficients $\vec{\alpha} = \{\alpha_P\}_{P \in \{I, X, Y, Z\}^{\otimes n}}$, we let $M(\alpha) = \sum_P \alpha_P P$.
For any Pauli string $P$, we let $|P|$ denote the degree of $P$, \textit{i.e.}, the number of non-identity terms in $P$.
It will also be useful to us to index by subsets of Pauli matrices, \textit{i.e.} $\vec{\alpha} = \{\alpha_P\}_{P \in S}$ for some subset $S$ of the Pauli matrices, in which case we let $M(\vec{\alpha}) = \sum_{P \in S} \alpha_P P$, and abbreviate it as $M$ when the context is unambiguous.

A particularly important type of Pauli matrix that will arise often throughout this paper will be the Pauli $Z$ matrices, and it will be convenient to index these by their non-identity parts.
To this end, for any $S \subseteq [n]$, we let $Z_S = \left( \otimes_{i \in S} Z_i \right) \bigotimes \left( \otimes_{i \not\in S} I \right)$\; .
We will also let $X_S$ be defined analogously for the Pauli $X$ matrices and for any $S \subseteq [n]$.


For two real quantities $x$ and $y$, the notation $x\lesssim y$ (resp. $x\gtrsim y$) means that there exists a universal constant $C>0$ independent from all parameters such that $x\leq Cy$ (resp. $x\geq Cy$).
The notation $x\lesssim_p y$ means $C=C(p)$ may only on the parameter $p$, and so on.

\subsection{Access models}

Here we formally define the class of algorithms we will consider throughout this paper, and some related parameters.
Throughout the paper, we will let $n$ denote the number of qubits of the Hamiltonian, and $k$ will denote the locality of the Hamiltonian.
We first recall the definition of a general, informationally complete measurement operator:
\begin{definition}[Positive operator-valued measure (POVM)]
    We say a collection of $d \times d$ PSD matrices $\mathcal{M} = \{M_i\}_{i = 1}^p$ with associated matrices $K_i \in \C^{d} \times \C^d$ form a \emph{positive operator-valued measure} (POVM) if $M_i = K_i K_i^\dagger$ for all $i = 1, \ldots, p$, and
    $ \sum_{i = 1}^p M_i = I$.
    We call $\{K_i\}_{i = 1}^p$ the associated \emph{Kraus operators}.

    Given a pure state $\ket{\psi} \in \C^{d}$, the outcome of measuring $\rho$ with $\mathcal{M}$ is that we obtain $i$ with probability $\bra{\psi} M_i \ket{\psi}$, for $i = 1, \ldots, p$, and given that the measurement returns $i$, the post-measurement state becomes $\ket{\psi'}$, where
    \[
    \ket{\psi'} = \frac{K_i^\dagger \ket{\psi}}{\norm{K_i^\dagger \ket{\psi}}} \; .
    \]
\end{definition}
\noindent
Note that in this definition, we never forget about the measurement outcomes, which will be without loss of generality for our setting.

We now define the first of model computation we will consider in this paper, which is also the model considered in prior work of~\cite{bluhm2024hamiltonian,huang2023learning}:
\begin{definition}[The sequential experiments Hamiltonian access model]
\label{def:access}
    We say an algorithm $\mathcal{A}$ is a sequential $R$-experiment algorithm for learning from real-time evolution that uses $r$ ancillae if, given an unknown $n$-qubit, $k$-local Hamiltonian $H$, for each round $\ell = 1, \ldots, R$, it operates as follows:
    \begin{enumerate}
        \item Planning: the algorithm chooses...\begin{itemize}
        \item a number of rounds $m^{(\ell)} \in \mathbb{N}$.
        \item timesteps $t^{(\ell)}_j \in \R_{\geq 0}$, 
        \item initial states $\ket{\psi^{(\ell)}} \in \C^{2^{n + r}}$, 
        \item unitaries $U^{(\ell)}_j \in \C^{2^{n + r}} \times \C^{2^{n + r}}$, and
        \item a POVM $\mathcal{M}^{(\ell)}_1$ on $n + r$ qubits.
    \end{itemize}
    \item It then forms the state given by the following real-time evolution:
    \begin{equation}
    \label{eq:evolution}
    \ket{\phi^{(\ell)}} = \prod_{j = 1}^{m^{(\ell)}} \left( e^{-iH t_j} \otimes I \right) U_j \ket{\psi^{(\ell)}} \, ,    
    \end{equation}
    and measures $\ket{\phi^{(\ell)}}$ with $\mathcal{M^{(\ell)}}$ to obtain post-measurement state $\ket{\psi^{(\ell+1)}}$
    \end{enumerate}
We allow all of these choices to be adaptive based on the outputs of the measurements from previous rounds.
If the number of rounds $R = 1$, we say the algorithm is \emph{non-adaptive}.

For any such algorithm:
\begin{itemize}
\item we say that the \emph{maximum time evolution} of the algorithm, denoted $\tmax$ is the maximum $t^{(\ell)}_j$ across all possible random choices of $t^{(\ell)}_j$, across all $\ell, j$,
\item we say that the \emph{minimum time resolution} of the algorithm, denoted $\tmin$ is the minimum $t^{(\ell)}_j$ across all possible random choices of $t^{(\ell)}_j$, across all $\ell, j$,
\item 
we say that the \emph{total rounds of interaction} of the algorithm, denoted $m$, is the maximum of $\sum_{\ell = 1}^R m^{(\ell)}$ over all possible random choices of $m^{(\ell)}$, and
\item we say that the \emph{total time evolution} of the algorithm, denoted $T$, is the maximum of $\sum_{\ell, j} t^{(\ell)}_j$ over all possible random choices of $t^{(\ell)}_j$.
\end{itemize}
\end{definition}
\noindent
We note that our results will be essentially agnostic to the amount of additional quantum memory, so for conciseness, when clear from context, we will often omit the choice of $r$.
Similarly, most of our results will also hold for arbitrary amounts of adaptivity, so unless otherwise specified, the reader should assume we are considering arbitrarily adaptive algorithms.

While this class of algorithms captures existing state-of-the-art algorithms such as those in~\cite{huang2023learning,bakshi2024learning,hu2025ansatz}, and previous lower bounds such as those in~\cite{bluhm2024hamiltonian,huang2023learning} also are against these sorts of algorithms, we note that they in fact do not capture all physically realizable quantum evolutions.
In particular, they do not capture algorithms which use intermediate partial measurements.
Instead, we observe that such algorithms are naturally captured by a learning tree, which can be thought of as the Hamiltonian analog of the learning tree formalism for quantum state learning introduced in~\cite{aharonov2022quantum,bubeck2020entanglement}, although we note that in contrast to their models, this model captures \emph{all} quantum learning algorithms, and not a restricted class of them:
\begin{definition}[General real-time evolution Hamiltonian access]
\label{def:general-access}
    A general algorithm $\mathcal{A}$ for learning from real-time evolutions of a Hamiltonian that uses $r$ ancillae is specified by a tree $\mathcal{T}$, and some initial state $\ket{\psi} \in \C^{2^{n + r}}$.
    Each leaf node $v$ is indexed by POVM $\mathcal{M}_v= \{M^{(v)}\}_{i = 1}^{p_v}$ and classical outputs $o_1, \ldots, o_{p_v}$, associated to each measurement outcome, and every internal node $v$ is indexed by:
    \begin{itemize}
        \item An evolution time $t_v$,
        \item A unitary matrix $U_v: \C^{2^{n + r}} \times \C^{2^{n + r}}$,
        \item a POVM $\mathcal{M}_v = \{M^{(v)}_i\}_{i = 1}^{p_v}$ over $n + r$ qubits with associated Kraus operators $\{ K_i \}_{i = 1}^{p_v}$, and
        \item $p_v$ child nodes, with an edge to each child indexed by a measurement outcome.
    \end{itemize}
    Any $n$-qubit Hamiltonian $H$ induces a distribution on the leaves of the tree by running the algorithm as follows.
    At every node $v$, it maintains some $(n + r)$-qubit state $\ket{\psi_v}$, initially $\ket{\psi}$.
    At an internal node $v$, it forms the state
    \[
    \ket{\phi_v} = \left((e^{-iHt_v} \otimes I) \right) U_v \ket{\psi_v}\; ,
    \]
    then measures the resulting state with $\mathcal{M}_v$.
    If it receives measurement outcome $i$, it transitions to the $i$-th child node $v'$, and the state $\ket{\psi_{v'}}$ is the post-measurement state of $\ket{\phi_v}$ with this measurement outcome.
    When the algorithm reaches a leaf node, it measures $\ket{\psi_v}$ using the POVM at that node, and based on the measurement outcome, outputs the associated classical output.

    For any learning tree $\mathcal{T}$, and any Hamiltonian $H$, we let $\mathcal{T}_H$ denote the distribution over leaf nodes induced by these dynamics when run with $H$.
    We say the learning tree has $m$ rounds of interaction if the depth of the tree is $m$, we say that it has maximum time evolution $\tmax = \max_{v} t_v$ and minimum time resolution $\tmin = \min_v t_v$.
    For any leaf node $v$, we say the total time evolution for that node, denoted $\tau_v$, is the sum of the time evolutions on the root-to-leaf path, and we say the total time evolution of the algorithm is 
    \[
    T = \max_{v~\mathrm{leaf}} \tau_v \; .
    \]
\end{definition}
\noindent
It is straightforward to see that this generalizes all previous formal models for learning from real-time evolution of Hamiltonians, such as those considered in~\cite{bluhm2024hamiltonian,huang2023learning}.
Additionally, it allows for essentially all physically realizable quantum operations, including several which were not captured by prior models, such as partial measurement, and controlled application of the Hamiltonian.

\subsection{Formal problem statements}

Here, we formally define the Hamiltonian learning problems we will consider throughout this paper.

\paragraph{Full parameter recovery}
The first problem we consider is the question of learning all of the parameters of a local Hamiltonian to good uniform accuracy.

\begin{definition}[Full parameter recovery for $k$-local Hamiltonians from time evolution]
\label{def:full-param-recovery}
    In the \emph{full parameter recovery problem}, an algorithm is given time evolution access (as in~\Cref{def:access}) to an unknown $k$-local, $n$-qubit Hamiltonian $H = \sum_{|P| \leq k} \alpha_P P$ satisfying $|\alpha_P| \leq 1$ for all $P$, and an accuracy parameter $\eps$.
    We say that the algorithm \emph{succeeds} if it outputs estimates $\{\widehat{\alpha}_P\}_{|P| \leq k}$ satisfying 
    \[
    |\alpha_P - \widehat{\alpha}_P| \leq \eps
    \]
    for all $P$ with degree at most $k$, with probability at least $2/3$.
\end{definition}
\noindent
As mentioned above, we note that some bound on the size of the coefficients is both necessary and standard.
We also note that by standard boosting arguments, the constant $2/3$ is arbitrary and can be replaced with any constant less than $1$ without changing the asymptotic complexity of the problem.

\paragraph{Last parameter recovery}
We next consider the (seemingly) much simpler problem, of learning a single parameter of a local Hamiltonian, given a description of the rest of the Hamiltonian:
\begin{definition}[Last parameter recovery for $k$-local Hamiltonians from time evolution]
\label{def:last-param-recovery}
    In the \emph{last parameter recovery problem}, an algorithm is given a $k$-local, $n$-qubit Pauli matrix $P$, a full description of a $k$-local Hamiltonian $H = \sum_{|Q|\leq k, Q\neq P} \alpha_Q Q$ satisfying $|\alpha_Q| \leq 1$ for all $Q$, an upper bound parameter $B$, and an error parameter $\eps>0$.
    The algorithm is then given time evolution access (as in~\Cref{def:general-access}) to the $k$-local $n$-qubit Hamiltonian $M = H + \beta P$ for some $|\beta| \leq B$.
    We say that the algorithm \emph{succeeds} if it outputs an estimate $\widehat{\beta}$ so that $|\widehat{\beta} - \beta| \leq \eps$, with probability at least $2/3$.
\end{definition}
\noindent
Once again, the constant $2/3$ is arbitrary and can be replaced with any constant less than $1$.
We note that in this definition, we allow the coefficient on $P$ to be much larger than the rest of the coefficients.
Our lower bounds will hold if we take $B = 1$, but will still hold even if we allow $B$ and $\eps$ to be very large.

\paragraph{Effective Hamiltonian recovery}
Finally, we also consider a ``non-parametric'' notion of Hamiltonian recovery, where the goal is not necessarily to recover the Pauli coefficients themselves, but to simply recover a functionally equivalent description of the Hamiltonian.
In other words, we can hope to (approximately) reproduce the unitary channel induced by the Hamiltonian.
Formally:
\begin{definition}[Unitary evolution]
For any unitary $V\in \C^{2^n\times 2^n}$, let
$\mathcal{U}:\mathbb{C}^{2^n\times 2^n} \to \mathbb{C}^{2^n\times 2^n}$ denote the associated unitary channel of $V$, $\mathcal{U}(V):\rho \mapsto V \rho V^{\dagger}$.
\end{definition}
For any linear map on complex matrices $\Phi: \C^{2^n \times 2^n} \to \C^{2^n \times 2^n}$, we let $\norm{\Phi}_{\diamond}$ denote the diamond norm of the map:
\[
\norm{\Phi}_{\diamond} = \sup_{\norm{M}_1 \leq 1} \norm{(\Phi \otimes I) M}_1 \; .
\]
It is well-known (see \textit{e.g.}~\cite{watrous2018theory}) that the diamond distance between two quantum channels governs the extent to which it is possible to discriminate between the channels.

The problem we consider in this paper is:
\begin{definition}[Effective Hamiltonian recovery]
\label{def:effective}
    In the \emph{effective Hamiltonian recovery problem}, an algorithm is given time evolution access (as in~\Cref{def:access}) to an unknown $k$-local, $n$-qubit Hamiltonian $H = \sum_{|P| \leq k} \alpha_P P$ satisfying $\norm{H}_\infty \leq 1$, a time parameter $t > 0$, and an accuracy parameter $\epsilon > 0$.
    We say that the algorithm \emph{succeeds} if it outputs a description of a unitary matrix $U \in \C^{2^n \times 2^n}$ so that $\norm{\mathcal{U}(U) - \mathcal{U}(e^{-iHt})}_{\diamond} \leq \eps$ with probability at least $2/3$.
\end{definition}
\noindent
As before, the constant $2/3$ is more or less arbitrary.
We believe this is the weakest possible notion of recovering a functional approximation of the Hamiltonian, as it asks for an approximation of the dynamics induced by the Hamiltonian.
In particular, we note that this is weaker than asking for \textit{e.g.} a spectral approximation to $H$ as in~\cite{bluhm2024hamiltonian} (for constant $t$).
We also note that the assumption that $H$ is spectrally bounded is a stronger assumption than the parameter-wise bound we assumed before; despite this, we will still show that this problem exhibits a strong lower bound.

\subsection{Fourier analysis preliminaries}
We will be using standard concepts in Boolean Fourier analysis.
An interested reader can find a more detailed introduction in~\cite{o2014analysis}.
For any set $A$, and any function $f: A \to \C$, we will let $\E_{X \sim A} [f(X)]$ denote the expectation of $f$ under a uniformly random draw from $S$.
For any subset $S \subseteq [n]$, we will let $\chi_S (x): \{ \pm 1\}^n \to \R$ be the Fourier character $\chi_S(x) = \prod_{i \in S} x_i$, and for any function $f: \{\pm 1\}^n \to \C$ we will let $\wh{f}(S) = \iprod{f, \chi_S}$ denote its corresponding Fourier coefficient, where 
\[\iprod{f, g} = \E_{X \sim \{\pm 1 \}^n} \left[ f(x) \overline{g(x)} \right] = 2^{-n} \sum_{x \in \{\pm 1\}^n} f(x) \overline{g(x)}\] 
is the standard (normalized) inner product.
For any $\chi_S$, we say that it has degree $d$ if $|S| = d$.
We say $f$ is a degree-$d$ polynomial if all of the nonzero Fourier coefficients of $f$ have degree at most $d$, and we let $\calP_{n, d}$ denote the set of degree-$d$ polynomials over the hypercube.

\subsection{Metrics}
For two classical probability distributions $\mathcal D_1,\mathcal D_2$, we write $d_\tv(\mc D_1,\mc D_2)$ to denote the total variation distance between them:
\[
d_\tv (\mc D_1, \mc D_2) = \sup_{E} \Pr_{\mc D_1} [E] - \Pr_{\mc D_2} [E] \; .
\]
\noindent
We will also need the following metric between unitaries, which (up to constant factors) captures the diamond distance between their associated channels:
\begin{definition}
    For unitaries $U,V\in \mathbb{C}^{2^n\times 2^n}$ we define
    \[
        \dph(U,V) = \min_{z\in \mathbb{C}:|z|=1} \norm{zU - V}_\infty
    \]
\end{definition}
\noindent
For this, we have:
\begin{lemma}[Equivalence of diamond norm and phase operator norm, \cite{haah2023query}]
For unitaries $U,V\in \mathbb{C}^{2^n\times 2^n}$ we have
\[
    \frac12 \norm{\mathcal{U}(U)-\mathcal{U}(V)}_\diamond \le \dph(U,V) \le \norm{\mathcal{U}(U)-\mathcal{U}(V)}_\diamond,
\] 

    
\end{lemma}

\subsection{Matrix calculus}
We will need the following chain rule formula for matrix-valued exponentials, due to Duhamel:
\begin{lemma}[Duhamel's formula,~\cite{duhamel1860elements}]
\label{lem:derivative} For any matrices $A$ and $V$, we have that
    \[
        \frac{d}{dt} e^{A + tV} = \int_0^1 e^{(1-\tau)(A + tV)} V e^{\tau (A + tV)} \:d\tau \; .
    \]
\end{lemma}

\section{Reducing hardness of learning to large Fourier-normalized absolute minima}


In this section we reduce learning lower bounds to the existence of polynomials with large Fourier-normalized absolute minima over the hypercube.

The main technical result of this section is the following reduction, which demonstrates that the existence of polynomials with large Fourier-normalized absolute minima immediately implies a lower bound against the last parameter recovery problem:
\begin{theorem}
\label{thm:reduction}
    Let $d \leq n -1$, let $\eta> 0$.
    Then, any algorithm for the last parameter recovery problem of $(d + 1)$-local Hamiltonians over $n$ qubits from time evolution (see \Cref{def:last-param-recovery}) with upper bound $B \geq \eta$ and error parameter $\eps = \eta/2$ with maximum time evolution $\tmax$ and $m$ total rounds of interaction requires
    \[
    \max(1,\eta \cdot \tau_{\max}) \cdot m \ge \frac{\max_{f \in \calP_{n - 1, d}} \gamma (f)}{\eta} \; .
    \]

\end{theorem}


\noindent
In Section~\ref{sec:large-absolute-minima} we give three constructions of functions with the following properties:

\begin{theorem}
\label{thm:function-constructions}
We have:
\begin{enumerate}[(i)]
    \item For all $\delta > 0$, there is a multilinear function $f: \{ \pm 1\}^n \to \R$ satisfying $\gamma (f) \geq 2^{(1/2 - \delta) n}$, for $n$ sufficiently large.\label{thm:f-general}
    \item For all $d \geq 2$, there is a degree-$d$ polynomial $f$ satisfying $\gamma (f) \geq \Omega (\sqrt{n / d})^{\lfloor d / 2\rfloor}$.\label{thm:f-low-degree-simple}
    \item For all $d \leq O(n^{1/3})$, there is a degree-$d$ polynomial $f$ satisfying
    \[
    \gamma (f) = \Omega \left(  \frac{1}{d^{5}} \cdot \left( \frac{n}{d} \right)^{\tfrac{d - 1}{2}} \right) \; .
    \]
    \label{thm:f-low-degree-hard}
\end{enumerate}
\end{theorem}
\noindent By combining these two results, we obtain our main lower bounds for the last parameter recovery problem:
\begin{corollary}
\label{cor:last-param-general}
    Let $\eta > 0$, and let $\delta > 0$.
    Then, for all $n$ sufficiently large, any algorithm for the last parameter recovery problem for $n$-qubit Hamiltonians from time evolution with upper bound $B \geq \eta$ and error parameter $\eta / 2$ requires 
    \[
    \max(1, \eta \cdot \tmax) \cdot m = \Omega (2^{(1/2 - \delta) n} / \eta)
    \]
\end{corollary}
\begin{corollary}
\label{cor:last-param-local}
    Let $\eta > 0$.
    Any algorithm for the last parameter recovery problem for $k$-local $n$-qubit Hamiltonians from time evolution with upper bound $B \geq \eta$ and error parameter $\eta / 2$ requires $\max (1, \eta \cdot \tmax) \cdot m = \Omega (g(k) / \eta)$, where
    \[
    g(k) = \begin{cases} \displaystyle 
  \frac{1}{k^5}\left( \frac{n}{k} \right)^{\frac{k-2}{2}} & \text{if } k \leq O(n^{1/3}), \\[1.1em]
 \displaystyle \left( \sqrt{\frac{n}{2k}} \right)^{\lfloor\frac{k-1}{2}\rfloor} & \text{otherwise}\,.
\end{cases}
    \]
\end{corollary}
\noindent
We pause here to make several comments about our results.
At a high level, our results show that either the number of queries to the Hamiltonian, or the maximum time resolution (and thus the evolution time) of the algorithm must scale polynomially in the total number of parameters in the system, just to learn a single parameter of the Hamiltonian, even if these other parameters are known to the algorithm.
In the setting where the minimum time resolution is not too small, for instance, when $\tmin = (n / k)^{-o(k)}$, our results imply that any algorithm must invariably pay total evolution time which is $(n / k)^{\Omega (k)}$.
This is a very realistic assumption in practice, since existing (and near-term) systems cannot achieve such precise control over the time resolution of the applied dynamics, see \textit{e.g.}~\cite{stilck2024efficient,bakshi2024structure,dutkiewicz2024advantage}.

Second, note that our results do not require any assumptions on the number of ancilla qubits or rounds of adaptivity.
Thus, our lower bounds hold against the strongest possible form of algorithm for learning from time evolution.

Finally, note that our result holds for all $\eta$.
There are two particularly interesting regimes: first, when $\eta$ is small, our result matches the Heisenberg-limited scaling of the known upper bounds of~\cite{huang2023learning,bakshi2024learning}.
Second, when $\eta$ is large---and indeed, we can take $\eta = (n/k)^{O(k)}$---our result demonstrates that even detecting the presence of an \emph{exponentially large} spike cannot be done efficiently.

\subsection{Proof of Theorem~\ref{thm:reduction}}

In this section, we describe the reduction that yields Theorem~\ref{thm:reduction}.
To this end, let $f: \{\pm 1\}^{n - 1} \to \R$ be a degree-$d$ polynomial which achieves the maximal Fourier-normalized minimum $\gamma(f)$.
Without loss of generality, by scaling, we can assume that $\max_{S} |\widehat{f} (S)| = 1$, so that $|f(x)| \geq \gamma (f)$ for all $x \in \{\pm 1\}^{n - 1}$.

Our construction will be as follows.
Let the ``spiked'' Pauli coefficient be $\Delta$, where 
\begin{equation}
\Delta = I_1 \otimes \cdots I_{n-1} \otimes X = \begin{pmatrix}
  0& 1 & & &\\
  1& 0& & & \\
  & & 0 & 1 & & \\
  & & 1 & 0 & & \\
  & & & &\ddots& & \\
  & & & & & 0&1\\
  & & & & &1&0
\end{pmatrix} \; .
\end{equation}
Note that $\Delta$ is a single, 1-local Pauli.

Now, given $f$ as above, let $g: \{\pm 1\}^n \to \R$ be defined by $g(x) = f(x_{-n}) (1 + x_n)$, where for any $x \in \{\pm 1\}^n$, we define $x_{-n} \in \{\pm 1\}^{n - 1}$ to be the $n - 1$ length vector consisting of the first $n - 1$ coordinates of $x$.
Note that for any set $S \subseteq [n]$, if $n \not\in S$, we have that
\begin{align*}
    \wh{g}(S) &= \E_{x \sim \{ \pm 1\}^n } \left[ f(x_{-n}) (1 + x_n) \chi_S (x) \right] = \E_{x \sim \{ \pm 1\}^n } \left[ f(x_{-n}) \chi_S (x) \right] = \widehat{f}(S) \; ,
\end{align*}
and on the other hand, if $n \in S$, we have
\begin{align*}
    \wh{g}(S) &= \E_{x \sim \{ \pm 1\}^n } \left[ f(x_{-n}) (1 + x_n) \chi_S (x) \right] = \E_{x \sim \{ \pm 1\}^n } \left[ f(x_{-n}) \chi_{S \setminus \{ n\}} (x) \right] = \widehat{f}(S \setminus \{ n\}) \; ,
\end{align*}
so in either case, we have that $|\wh{g} (S)| \leq 1$.


We now define our $\vec{\alpha}$.
In fact, it will be supported only on the Pauli $Z$ matrices (\textit{i.e.} it is diagonal), and for any $Z_S$, we let $\alpha_{Z_S} = \wh{g} (S)$.
By the above, we have that $\norm{\vec{\alpha}}_\infty \leq 1$ for all $P$.
Note that by definition, for any computational basis vector $\ket{x}\in \{\pm 1\}^n$ we have that 
\[
\bra{x} M(\vec{\alpha}) \ket{x} = \sum_{S} \wh{g}(S) \bra{x} Z_S \ket{x} = \sum_{S} \wh{g} (S) \chi_S (x) = g(x) \; .
\]
With this, we can now state the main bound we will require:

\begin{theorem}\label{thm:worst-hard-instance}
    Let $\eta > 0$.
    Then, for all $t \geq 0$, we have that 
    \[
    \left\| e^{-iM(\vec{\alpha})t} - e^{-i(M(\vec{\alpha})+ \eta \Delta)t}\right\|_\infty \le O \left( \frac{\max(\eta,\eta^2 t)}{\gamma (f)} \right) \; .
    \]
\end{theorem}
\noindent
We first show~\Cref{thm:reduction} given~\Cref{thm:worst-hard-instance}.



\begin{proof}[Proof of~\Cref{thm:reduction} given~\Cref{thm:worst-hard-instance}]

    For simplicity throughout this proof, we abbreviate $M(\alpha)$ by $M$.
    Let $\mathcal{T}$ denote a learning tree with parameters $\eta$ and $\tmax$ as in Theorem~\ref{thm:reduction}, and let $h = \max(\eta,\eta^2 \tau_{\max})$.
    %
    %
    We will show the induced distribution over leaves of $\mathcal{T}$ under $M$ and $M+\eta \Delta$ differs by at most $O(mh/\gamma(f))$ in total variation distance.
    This immediately implies that $\mathcal{T}$ cannot distinguish between $M$ and $M + \eta \Delta$ except with probability at most $O(mh/\gamma(f))$, which implies that any successful algorithm for last parameter recovery, and hence for distinguishing these two Hamiltonians requires $m\cdot h \ge \Omega(\gamma(f))$, as claimed.
    
    We prove this by a hybrid argument. 
    Let $\mathcal{D}_k$ be the induced distribution over leaf nodes of the tree $\mathcal{T}$ given by the evolution where for a node $v$ with depth $< k$, we apply time evolution with Hamiltonian $M + \eta \Delta$:
    \[
    \ket{\phi_v} = \left((e^{-i(M+\eta \Delta)t_v} \otimes I) \right) U_v \ket{\psi_v}\; ,
    \]
    and for node $v$ with depth $\ge k$, we apply time evolution with Hamiltonian $M$:
    \[
    \ket{\phi_v} = \left((e^{-iMt_v} \otimes I) \right) U_v \ket{\psi_v}\;.
    \]

    Now, let $v$ be any depth-$k$ node, and let $\mathcal{D}_{k, v}$ and $\mathcal{D}_{k + 1, v}$ denote the conditional distributions of $\mathcal{D}_k$ and $\mathcal{D}_{k+1}$ conditioned on the event that they arrive at this node.
    The dynamics of $\mathcal{D}_k$ and $\mathcal{D}_{k + 1}$ are identical up to $v$, and so the associated state $\ket{\psi_v}$ is the same for both distributions, and the probability we reach that state, which we denote $p_v$, is also the same for both distributions.


    There is a difference of the evolution of the state $\ket{\psi_v}$: the evolved states of $\mathcal{D}_{k,v}$ and $\mathcal{D}_{k+1,v}$ are
     $\ket{\phi_v}$ and $\ket{\phi_v'}$ respectively, where \[
    \ket{\phi_v} = \left((e^{-iMt_v} \otimes I) \right) U_v \ket{\psi_v},\; \ket{\phi_v'} = \left((e^{-i(M+\eta \Delta)t_v} \otimes I) \right) U_v \ket{\psi_v}\; .
    \] 

    Then the dynamics of $\mathcal{D}_k$ and $\mathcal{D}_{k + 1}$ are identical again after this evolution of the node $v$, \textit{i.e.}, they are running the same POVM for states $\ket{\phi_v}$ and $\ket{\phi_v'}$, and $\mathcal{D}_{k, v}$ and $\mathcal{D}_{k + 1, v}$ are exactly the associated distributions of measurement outcomes.

    By \Cref{thm:worst-hard-instance}, \[\norm{\ket{\phi_v} - \ket{\phi_v'} }_2 \le O \left( \frac{\max(\eta,\eta^2 t_v)}{\gamma (f)} \right) \le O(h/\gamma(f)),\] which implies that $\norm{\ketbra{\phi_v}{\phi_v} -\ketbra{\phi_v'}{\phi_v'}}_1 \le O(h/\gamma(f))$. 
    Hence the total variation distance between $\mathcal{D}_{k,v}, \mathcal{D}_{k+1,v}$ can be upper bounded by $d_\tv (\mathcal{D}_{k,v}, \mathcal{D}_{k+1,v}) \le O(h/\gamma(f))$, for any node $v$ at depth $k$. 

    Note that $\mathcal{D}_k = \sum_{v:\mathrm{at~depth}~k} p_v\cdot \mathcal{D}_{k,v}$ , and $ \mathcal{D}_{k+1} = \sum_{v:\mathrm{at~depth}~k} p_v\cdot \mathcal{D}_{k+1,v}$, and so we have
    \[
        d_\tv (\mathcal{D}_k, \mathcal{D}_{k+1}) \le \sum_{v:\mathrm{at~depth}~k} p_v \cdot d_\tv(\mathcal{D}_{k,v}, \mathcal{D}_{k+1,v}) \le O(h/\gamma(f)). 
    \]

    Notice that induced distribution over leafs of $\mathcal{T}$ under $M$ and $M+\eta \Delta$ are exactly $\mathcal{D}_0$ and $\mathcal{D}_m$, so by triangle inequality, we have
    \[
        d_{\tv}(\mathcal{D}_0, \mathcal{D}_{m}) \le \sum_{k=0}^{m-1} d_{\tv}(\mathcal{D}_k, \mathcal{D}_{k+1}) \le O(mh/\gamma(f)),
    \] as we desired.%
    %
%
\end{proof}

    
\noindent
The rest of the section is dedicated to the proof of Theorem~\ref{thm:worst-hard-instance}.

\begin{proof}[Proof of Theorem~\ref{thm:worst-hard-instance}]
Note that $M$ and $M(\vec{\alpha}) + \eta \Delta$ are both block diagonal with $2 \times 2$ blocks, where each block is supported on the vectors $\ket{x \oplus 1}, \ket{x \oplus (-1)}$ for all $x \in \{\pm 1\}^{n - 1}$.
For any $x \in \{ \pm 1\}^{n - 1}$, let $M_x$ and $\Delta_x$ denote the $2 \times 2$ matrices that are $M(\vec{\alpha})$ and $\Delta$ restricted to these indices, respectively.
By the shared block diagonal structure, we have that
\[
\norm{e^{-it M(\vec{\alpha})} - e^{- i t (M(\vec{\alpha}) + \eta \Delta)}}_\infty \leq \max_{x \in \{ \pm 1\}^{n - 1}} \left\| e^{-it M_x} - e^{- it(M_x+\eta \Delta_x)}\right\|_\infty \; .
\]
For any such $x$, we have that $M_x$ is a diagonal matrix with entries $2 f(x)$ and $0$ by recalling the definition of $M$ and function $g$, and $\eta \Delta_x$ is a matrix with zero diagonal entries and spectral norm at most $\eta$.
The theorem then follows from Lemma~\ref{lem:matrix-bound} (by taking $A$, $\Delta$, $C$ and $D$ to be $-Mt$, $-\eta t \Delta$, $\eta t$, and $\gamma(f)t$ respectively), proved below.
\end{proof}

\begin{lemma}
\label{lem:matrix-bound}
    Let $A \in \C^{\ell \times \ell}$ be a diagonal matrix, and let $\Delta \in \mathbb{C}^{\ell \times \ell}$ be Hermitian. 
    Let $D$ be any positive constant.
    Suppose that $M$ satisfies $|A_{jj} - A_{kk}| \geq D$ for all $j \neq k$, $\Delta_{jj} = 0$ for all $j$, and $\opnorm{\Delta} \leq C$.
    Then, we have that
    \[
    \opnorm{ e^{iA} - e^{i(A+\Delta)}} \le O\left(\min \left( \frac{\max(C,C^2) \ell}{D},C,1\right) \right) \; .
    \]
\end{lemma}
\begin{proof}
By the triangle inequality, the left hand side is bounded by $2$.
Thus, by Lemma~\ref{lem:derivative}, and another application of the triangle inequality, it suffices to show that 
\begin{equation}
\label{eq:derivative1}
   \opnorm{\int_0^1 e^{i(1-\tau)(A+t\Delta)} \Delta e^{i\tau (A+t\Delta)} \:d\tau} \leq O\left(\min \left( \frac{\max(C,C^2) \ell}{D},C\right) \right) \; ,
\end{equation}
for all $t \in [0, 1]$. 
Let $B_t (\tau)$ denote the integrand on the LHS of~\Cref{eq:derivative1}.
$B_{t} (\tau)$ is clearly spectrally upper bounded by $C$ for all $t, \tau$, and so the overall integral is also clearly at most $C$, for all $D$.
Thus, it remains to consider the setting where $D \geq c' C \ell$ for some constant $c'$ sufficiently large. 
In this setting, by the triangle inequality, it suffices to demonstrate that, for any fixed $t \in [0, 1]$, the LHS of~\Cref{eq:derivative1} is at most $O(C \ell/D)$.

To that end, for any $t \in [0, 1]$, let $\sum_{j = 1}^n \lambda_j (t) \ketbra{u_j(t)}{u_j(t)}$ denote the spectral decomposition of $A + t \Delta$, where $\lambda_1 (t) \geq \lambda_2(t) \geq \ldots \geq \lambda_n (t)$.
Note that $u_j (0) = \ket{j}$ and that $\lambda_j (0) = A_{jj}$.
By Weyl's inequality, observe that $\left| \lambda_j (t) - A_{jj} \right| \leq C$ for all $j$, and hence $\left| \lambda_j (t) - \lambda_k (t) \right| \geq \Omega (D)$, for all $j \neq k$.

We will show that for all $j, k = 1, \ldots, \ell$,
\begin{equation}
    \label{eq:derivative-entrywisebound}
    \left| \int_{0}^1 \bra{u_j(t)} B_t (\tau) \ket{u_k(t)} d \tau \right| \leq O \left( \frac{\max(C,C^2)}{D} \right) \; .
\end{equation}
Since the $u_j$ form an orthonormal basis, this immediately implies~\Cref{eq:derivative1}.

First, consider the case where $j \neq k$.
In this case, we have that 
\begin{align*}
    \left| \int_0^1 \bra{u_j(t)} B_t (\tau) \ket{u_k(t)} d \tau \right| &= \left| \bra{u_j (t)} \Delta \ket{u_k (t)} \right| \cdot \left| \int_0^1 e^{i((1-\tau)\lambda_k (t) + \tau\lambda_j (t))} d \tau \right| \\
    &= \left| \bra{u_j (t)} \Delta \ket{u_k (t)} \right| \cdot \left| \frac{e^{i\lambda_j (t)} - e^{i\lambda_k(t)}} {\lambda_j(t) - \lambda_k(t)} \right| \\
    &\leq O\left( \frac{C}{|\lambda_j(t) - \lambda_k(t)|} \right) \leq O \left( \frac{C}{D} \right) \; ,
\end{align*}
as claimed, where the third line follows since $\opnorm{\Delta} \leq C$.

We now consider the case where $j = k$.
Here, we have that
\begin{align*}
\left| \int_0^1 \bra{u_j(t)} B_t (\tau) \ket{u_k(t)} d \tau \right| &= \left| \int_0^1 \bra{u_j (t)} e^{-i (1 - \tau) (A + t \Delta)} \Delta e^{-i \tau (A + t \Delta)} \ket{u_j (t)} dt \right| \\
&\le
 \int_0^1 \left| \bra{u_j (t)} \Delta  \ket{u_j (t)} \right| dt  \\
&=
\left| \bra{u_j (t)} \Delta \ket{u_j (t)} \right|\; ,
\end{align*}
since $u_j$ is an eigenvector of $A + t \Delta$ by definition.

Write $u_j (t) = c_j \ket{j} + \beta_j \ket{\epsilon_j}$, where $\epsilon_j$ is the component of $u_j (t)$ orthogonal to $\ket{j}$.
We claim that $|\beta_j|^2 \leq O(C^2/D^2)$.
To see this, let $\beta_j = \sum_{k \neq j} c_k \ket{k}$.
Because $u_j (t)$ is a eigenvector of $A + t \Delta$ with eigenvalue $\lambda_j (t)$, by rearranging, we obtain that
\begin{align*}
    \Delta \ket{u_j (t)} = \frac{1}{t} \cdot \left( (\lambda_j (t) - A_{jj}) c_j \ket{j} + \sum_{k \neq j} (\lambda_j (t) - A_{kk}) c_k \ket{k} \right) \; .
\end{align*}
The LHS of this expression has norm at most $C$ since $\Delta$ has spectral norm at most $C$, and hence
\begin{align*}
    C^2 &\geq \frac{1}{t^2} \cdot \left( \sum_{k \neq j} (\lambda_j (t) - A_{kk})^2 |c_k|^2 \right) \\
    &\geq \sum_{k \neq j} \Omega (D)^2 |c_k|^2 \; ,
\end{align*}
from which we derive that $\sum_{j \neq k} |c_k|^2 \leq O(C^2 / D^2)$, which is equivalent to the desired statement.
Since $\bra{j} \Delta \ket{j} = \Delta_{jj} = 0$, we have that
\begin{align*}
     \left| \bra{u_j (t)} \Delta \ket{u_j (t)} \right| &= \left| \left(c_j \bra{j} + \beta_j \bra{\epsilon_j} \right) \Delta \left(c_j \ket{j} + \beta_j \ket{\epsilon_j} \right) \right| \\
     &\leq 2 |\beta_j c_j \bra{\epsilon_j} \Delta \ket{j}| + \beta_j^2 |\bra{\epsilon_j} \Delta \ket{\epsilon_j}| \\
     &= O \left( \frac{C^2}{D} \right) + O \left( \frac{C^3}{D^2} \right) = O \left( \frac{C^2}{D} \right) \; ,
\end{align*}
as claimed.
\end{proof}

%% file: instance.tex
\section{Polynomials with large absolute minima}
\label{sec:large-absolute-minima}

In this section we exhibit classes of polynomials $f:\{\pm 1\}^n\to \R$ with large Fourier-normalized absolute minima.
One may imagine this task as maximizing the absolute minimum of $f$ given the requirement that the coefficients of $f$ remain in $[-1,1]$.

We begin by presenting two straightforward constructions: a probabilistic argument for general multilinear $f$, and an explicit construction for degree-$d$ $f$ which works for all pairs $(d,n)$.
However, in the important regime of $\deg(f)\leq O(n^{1/3})$ (including for example constant-degree $f$), it turns out this latter construction is suboptimal; $\gamma(f)$ of such loses significantly in the exponent on $n$ here.
In \Cref{sec:krawtchouk} we derive stronger lower bounds in the regime of $\deg(f)\leq O(n^{1/3})$.

First, a strong lower bound on $\gamma(f)$ for general $f$:
\begin{lemma}
\label{lem:boolean-instance}
    For all constant $\delta > 0$, and for all $n \geq \Omega (1 / \delta)$, there is a function $f: \{\pm 1\}^n \to \R$ such that $|\wh{f}(S)| \leq 1$ for all $S$, and $|f(x)| \geq 2^{(1/2-\delta)n}$ for all $x \in \{\pm 1\}^n$. 
\end{lemma}
\begin{proof}
Let $f$ be a random boolean function, \textit{i.e.} $f(x) = z_x$, where $z_x$ is uniformly drawn from $\{\pm 2^{(1/2-\delta)n}\}$ independently, for all $x \in \{\pm 1\}^n$.
We will show that such an $f$ will have the desired properties with high probability.
Note that for any $S\subseteq [n]$, 
\begin{align*}
    \wh{f}(S) = 2^{-n} \textstyle\sum_{x \in \{\pm 1\}^n} \chi_S (x) \cdot z_x
\end{align*}
is an average of $2^n$ independent random variables with coefficients bounded by $2^{(1/2-\delta)n}$, so by Hoeffding's inequality, 
\[
    \Pr_f[|\wh{f} (S)| > 1] \le \exp(-\Omega(2^{2\delta n})).
\]
Hence by a union bound over all $S\subseteq [n]$, we know with probability $1-\exp(-\Omega(2^{2\delta n}))$, $|\wh{f}(S)|\le 1$ for all $S$, so the lemma follows.
\end{proof}

When we want to construct \textit{low-degree} $f$ with large absolute minima, it is less clear how to make a probabilistic construction go through.
However, it turns out we can get an explicit family by making use of an elementary number-theoretic observation.
Consider $f(x)=\sum_{i<j}x_ix_j = \big((\Sigma_ix_i)^2-n\big)/2$.
Now $(\Sigma_ix_i)^2$ is always a perfect square, but for infinitely-many $n$, $n$ is $\Omega(\sqrt{n})$-far from any perfect square.
Thus for these $n$ we have $\min_x|f(x)|\geq \Omega(\sqrt{n})$.
Tensorizing this idea yields the following.
\begin{lemma}
\label{lem:local-boolean-instance}
    Let $n$ be sufficiently large.
    Then, there is a degree $d \ge 2$ function $f: \{\pm 1\}^n \to \R$ such that $\|\wh f\|_\infty\leq 1$ and $|f(x)| \geq \Omega(\sqrt{n/(2d)}^{\lfloor d/2 \rfloor})$ for all $x \in \{\pm 1\}^n$. 
\end{lemma}

\begin{proof}
    Let $s = \lfloor \sqrt{n/d} \rfloor$ and $\ell = s^2 + s$. Since $(d/2) \cdot \ell \le n$, we can make $\lfloor d/2 \rfloor$ disjoint sets $S_1, S_2 \dots S_{d/2} \subset [n]$, with the same size $\ell$.

    Let \[
        g_t(x) = \sum_{\{i,j\}\subset S_t} x_i x_j = \frac{\big(\textstyle\sum_{i\in S_t} x_i\big)^2 - \ell}{2}, 
    \]
    and notice that $\sum_{i\in S_t} x_i \in \mathbb{Z}$, and since the distance of $s^2 + s$ to any perfect square is at least $s$, we have $|g_t(x)| \ge s$, for all $x\in \{\pm 1\}^n$.

    Let $f$ be the following:
    \[
        f(x) := \prod_{t=1}^{d/2} g_t(x),
    \] which is a degree-$d$ Boolean function with coefficients in $\{0,1\}$.
    Hence for all $x\in \{\pm 1\}^n$,
    \[
        |f(x)| = \prod_{t=1}^{\lfloor d/2\rfloor} |g_t(x)| \ge s^{\lfloor d/2\rfloor} = \lfloor \sqrt{n/d} \rfloor^{\lfloor d/2\rfloor} \ge \sqrt{n/(2d)}^{\lfloor d/2\rfloor}\,.\qedhere
    \]
\end{proof}

\subsection{Large absolute minima for $d\leq O(n^{1/3})$}
\label{sec:krawtchouk}

In this section we exhibit an (essentially) explicit class of degree-$d$ polynomials $f$ with $\gamma(f)$ larger than in \Cref{lem:local-boolean-instance}, in the setting of $d$ not too large.
Instead of generalizing the degree-two example by tensorization, here we generalize by considering higher-degree elementary symmetric polynomials.
To that end, let $e_d^{(n)}=\sum_{S\subseteq[n]:|S|=d}\chi_S$, the multilinear degree-$d$ elementary symmetric polynomial on $n$ variables.
Then we have the following family with large absolute minima.

\begin{theorem}
    \label{thm:high-disc-poly}
    There exists a universal $c>0$ and $N\in \mathbb N$ such that following holds.
    For all $n\geq N$ and any $d\leq cn^{1/3}$, there exists a degree-$d$ polynomial $f_{n,d}:\{\pm 1\}^n\to \R$ with $\|\wh{f_{n,d}}\|_\infty\leq 1$ such that
    \[\min_{x\in\{\pm 1\}^n} |f_{n,d}(x)|\;\geq\; \Omega\!\left(\frac{1}{d^{5}}\left(\frac{n}{d}\right)^{\frac{d-1}{2}}\right).\]
    In particular, $f_{n,d}=e^{(n^*)}_d$ for some $n^*\in[n-O(\sqrt{n/d}),n]$.
\end{theorem}

Note that with $t=t(x)=|\{j:x_j=-1\}|=(n-\sum_jx_j)/2$, we have
    \[e_d(x)=\sum_{j=0}^d(-1)^j\binom{t}{j}\binom{n-t}{d-j}=K_d(t;n),\]
the (binary) Krawtchouk polynomials.
The main idea of our argument is to find an $n^*$ not too far from $n$ for which all the roots of $K_d(\,\cdot\,;n):\R\to\R$ are $\Omega_d(1)$-far from $\Z$.
Combining this with the fact that $K_d(\,\cdot\,;n)$ is degree-$d$ and has well-spaced roots, we conclude that $K_d(\,\cdot\,;n)$ is sufficiently large at every integer, and hence $e^{(n^*)}_d$ has large absolute minimum on $\{\pm 1\}^n$.

To find a suitable $n^*$, we will actually track the $j$\textsuperscript{th} root of $K_d(\,\cdot\,;n)$ as $n$ changes, show it is often far from $\Z$, and then union bound over the $d$-many roots in $K_d(\,\cdot\,;n)$ to find an $n^*$ where all roots are simultaneously far from $\Z$.

To execute this approach, the key technical fact we must argue is that as $n$ changes, the $j$\textsuperscript{th} root shifts in a fixed direction by more than 1, but not too much more: precisely with asymptotic behavior $1+\Theta_d(n^{-1/2})$.
We will turn to this root increment lemma next.

The proof of \Cref{thm:high-disc-poly} will in several places require some basic facts about Krawtchouk polynomials; for details we direct the reader to \cite[\S9.11]{Koekoek2010}.
We will also need a result of Krasikov and Zarkh \cite{KRASIKOV2009121} on the root spacing of $K_d(t;n)$.
With $\Delta_\text{min}=\Delta_\text{min}(m,d)$ denoting the smallest distance between roots of $K_d(\,\cdot\,; m)$, we have the following.
\begin{fact}[{\cite[Corollary 1]{KRASIKOV2009121}}]
\label{fact:root-space}
    The minimum distance between roots of $K_d(\,\cdot\,;m)$ is bounded as $\Delta_\textnormal{min}\geq\sqrt{2m/d}$ provided $d\leq m/2$.
\end{fact}

\begin{lemma}[Root Increment Lemma]
    \label{prop:root-inc}
    Let $\xi_j^{(m)}$ be the $j$\textsuperscript{th} root of $K_d(\,\cdot\,;m)$ greater than $m/2$.
    Then for all $m\in [n/2,n]$ and assuming $2d+2\leq m$,
    \[\xi_j^{(m)}-\xi_j^{(m-2)}\in 1+\left(\Omega\left(\frac{1}{\sqrt{dn}}\right),O\left(\sqrt{\frac{d}{n-2d}}\right)\right)\]
\end{lemma}

The proof of the root increment lemma relies on the theory of orthogonal polynomials, for which \cite[Chapter 1]{Gautschi2004-il} is a good reference.
We will also need a standard fact about analytic perturbations of symmetric matrices.

\begin{fact}[cf.~{\cite[Chapter II, Theorems 5.4 and 6.8]{Kato1995}}]
\label{fact:mat-perturb}
Let $A(u)$ be a $C^1$ family of real symmetric matrices for $u$ in an interval $I$.
Then the eigenvalues of $A(u)$, counted with multiplicity, admit $C^1$
labelings $\mu_1,\dots,\mu_N:I\to\R$.
Moreover, if for some $k$ the branch $\mu_k(u)$ is simple for every $u\in I$,
then for every $u\in I$ and every unit eigenvector $v$ of $A(u)$ with
eigenvalue $\mu_k(u)$,
\[
\mu_k'(u)=\langle v,A'(u)v\rangle.
\]
\end{fact}

\medskip
\begin{proof}[Proof of \Cref{prop:root-inc}]
It is convenient to consider the orthonormalized Krawtchouk polynomials 
\[p_j(t;m):=\frac{K_j(t;m)}{\sqrt{\binom{m}{j}}}\,,\]
which are orthonormal with respect to the binomial measure $\mu_m(t)=2^{-m}\binom{m}{t}$ \cite[\S9.11]{Koekoek2010}.
These satisfy the recurrence
\[tp_j(t;m)=a_{j+1}(m)p_{j+1}(t;m)+\frac m2p_j(t;m)+a_j(m)p_{j-1}(t;m)\]
\[\text{with}\qquad a_0=0\qquad\text{and}\qquad a_j(m)=\frac12\sqrt{j(m-j+1)},\quad j\geq 1.\]
It is a standard fact in the theory of orthogonal polynomials (see \textit{e.g.,} \cite[Theorem 1.31]{Gautschi2004-il}) that the roots of $p_d(\,\cdot\,;m)$ (and consequently those of $K_d(\,\cdot\,;m)$) are the eigenvalues of the Jacobi matrix
\[J_d(m) :=  \begin{pmatrix}
\frac m2 & a_1(m) &  &  &  \\
a_1(m) & \frac m2 & \ddots &  &  \\
 & \ddots & \ddots & a_{d-2}(m) &  \\
 &  & a_{d-2}(m) & \frac m2 & a_{d-1}(m) \\
 &  &  & a_{d-1}(m) & \frac m2 
\end{pmatrix}\,.\]
Then $J_d(m)=:\frac m2 I + T_d(m)$, and
\[\xi_j^{(m)}=\frac m2+\lambda_j^+(m),\]
where $\lambda_j^+(m)$ is the $j$\textsuperscript{th} positive eigenvalue of $T_d(m)$.

\medskip

\noindent\textbf{Upper bound.}
Using that the operator norm is upper-bounded by the maximum row sum, we get
\[\|T_d(m)-T_d(m-2)\|_\text{op}\leq \max_{1\leq j\leq d-1}2\big(a_j(m)-a_j(m-2)\big).\]
Now
\[a'_k(u)=\frac14\sqrt{\frac{k}{u-k+1}},\]
so there is some $u^*\in[m-2,m]$ such that
\[a_j(m)-a_j(m-2)=2a_j'(u^*)\leq \frac12\sqrt{\frac{j}{(m-2)-j+1}}\leq \frac12\sqrt{\frac{d}{m-d-1}}\leq \frac12\sqrt{\frac{2d}{n-2d}}\]
Combining these last displays gives
\[\|T_d(m)-T_d(m-2)\|_\text{op}\leq \sqrt{\frac{2d}{n-2d}}\,.\]
Thus by Weyl's inequality,
\[\xi_j^{(m)}-\xi_j^{(m-2)}=1+\big(\lambda_j^+(m)-\lambda_j^+(m-2)\big)\leq 1+\|T_d(m)-T_d(m-2)\|_\text{op}\leq 1+\sqrt{\frac{2d}{n-2d}},\]
finishing the upper bound on the finite difference interval.

\medskip
\noindent\textbf{Lower bound.}
Fix a positive eigenvalue $\lambda(u)=\lambda_j^+(u)$ of $T_d(u)$ and let $v=v(u)$ be a corresponding unit eigenvector.
Set
\[
w_k:=2a_k(u)v_kv_{k+1},
\qquad
c_k(u):=\frac{a_k'(u)}{a_k(u)}=\frac1{2(u-k+1)}.
\]
Observe that for all $k$,
\begin{equation}
\label{eq:a-ratio}
c_k(u)=\frac{a_k'(u)}{a_k(u)}\geq \frac1{2u};
\end{equation}
let us show that also
\begin{equation}
\label{eq:lambda-ratio}
    \frac{\lambda'(u)}{\lambda(u)}\geq \frac1{2u}.
\end{equation}

From the definitions so far,
\[
\lambda(u)=\langle v,T_d(u)v\rangle=\sum_{k=1}^{d-1} w_k,
\]
We would like to differentiate $\lambda(u)$ for
$u\in[m-2,m]$.
Note that for all such $u$ the off-diagonal entries of $T_d$ are nonzero, and so $T_d(u)$ has simple spectrum for $u\in[m-2,m]$ \cite{SymmEigs}.
Moreover, $T_d(u)$ is a $C^1$ family of real symmetric matrices, so \Cref{fact:mat-perturb} implies that
$\lambda(u)$ is a $C^1$ function on $[m-2,m]$.
Hence, for each $u\in[m-2,m]$ and any unit eigenvector $v$ of $T_d(u)$ with
eigenvalue $\lambda(u)$,
\begin{equation}
    \label{eq:eig-deriv}
    \lambda'(u)=\langle v,T_d'(u)v\rangle
    =\sum_{k=1}^{d-1} c_k(u)\,w_k .
\end{equation}


Also, from the eigenvalue equation \(T_d(u)v=\lambda v\),
\[
w_{2r-1}+w_{2r}=2\lambda v_{2r}^2\ge0,
\qquad
w_{2r}+w_{2r+1}=2\lambda v_{2r+1}^2\ge0,
\]
and at the boundary,
\[
w_{d-1}=2\lambda v_d^2\ge0,
\]
so every tail sum \(\sum_{k=\ell+1}^{d-1}w_k\) is nonnegative and we may apply summation by parts to \eqref{eq:eig-deriv}:
\[
\lambda'
=
c_1(u)\sum_{k=1}^{d-1}w_k
+
\sum_{\ell=1}^{d-2}\big(c_{\ell+1}(u)-c_\ell(u)\big)\sum_{k=\ell+1}^{d-1}w_k.
\]
Using that $c_k(u)$ is increasing in $k$, we get
\begin{equation}
    \label{eq:lambda-ratio-rearr}
    \lambda'\geq c_1\sum_{k=1}^{d-1}w_k=\frac{1}{2u}\lambda,
\end{equation}
which is \eqref{eq:lambda-ratio}.
Equivalently,
\[
\frac{d}{du}\log \lambda_j^+(u)\ge \frac{1}{2u}.
\]
Integrating from \(u=m-2\) to \(u=m\) gives
\[
\lambda_j^+(m)\ge \sqrt{\frac{m}{m-2}}\;\lambda_j^+(m-2).
\]
Thus it remains to lower bound $\lambda^+_j(m-2)$ uniformly in $j$, which is equivalent to getting a lower bound on the smallest positive eigenvalue of $T_d(m-2)$.
But the smallest positivie eigenvalue of $T_d(m-2)$ is nothing but the distance of the first root of $K_d(\,\cdot\,;m-2)$ past $(m-2)/2$.
The roots of $K_d(\,\cdot\,;m-2)$ are symmetric about $(m-2)/2$, so by \Cref{fact:root-space}, we get a lower bound of half the spacing:
\[\lambda^{+}_j(m-2)\geq \Delta_\text{min}(m-2,d)/2\gtrsim\sqrt{m/d}\,.\]

In conclusion, we obtain
\[\lambda_j^+(m)-\lambda_j^+(m-2)\geq \left(\sqrt{\frac{m}{m-2}}-1\right)\lambda_j^+(m-2)\gtrsim\sqrt{\frac{1}{md}}\geq \frac{1}{\sqrt{nd}}\,.\]

The lemma is proved by substituting back into $\xi_j^{(m)}-\xi_j^{(m-2)}=1+(\lambda_j^+(m)-\lambda_j^+(m-2))$.
\end{proof}

\begin{corollary}
    \label{cor:z-dist}
    There exist absolute constants $B,c>0$ such that for all $d\leq cn^{1/3}$, there is an (odd) $n^*\in(n-B\sqrt{n/d},n)$ such that all roots of $K_d(\,\cdot\,;n^*)$ are at least $\Omega(1/d^2)$-away from $\Z$.
\end{corollary}
\begin{proof}
    When $m$ is even and $d$ is odd, the Krawtchouk polynomial $K_d(t;m)$ has a root at $m/2\in \Z$, so to avoid case analysis we restrict to odd $m$, denoted by $m\in I_\mathsf{odd}:=(2\Z+1)\cap [n-B\sqrt{n/d},n]$.
    Note that $K_d(t;m)$ is symmetric about $m/2$, which for odd $m$ is always a half-integer.
    It thus suffices to argue the corollary only for roots past $m/2$.
    To that end fix the $j$'th such root. 
    
    By the upper bound in the Root Increment Lemma (\Cref{prop:root-inc}) as $m\in I_\mathsf{odd}$ goes from $n-B\sqrt{n/d}$ to $n$, the $j$\textsuperscript{th} root $\xi_j^{(m)}$ crosses 0 mod 1 at most $O(B\sqrt{n/d}\sqrt{d/n})=O(B)$ times.
    Let $r_d>0 $ denote a small radius about $0$.
    From the increment lower bound, we get that each time we pass $0$, the number of $\xi_j^{(m)}$ within $r_d$ of $0$ is at most $\max\{1,O(r_d/(1/\sqrt{dn}))\}=\max\{1,O(r_d\sqrt{dn})\}$.
    As a result, among $m\in I_\mathsf{odd}$, the number of roots $\{\xi_j^{(m)}\}_{m\in I_\mathsf{odd}}$ within $r_d$ of 0 is at most
    \[\max\{1,O(r_d\sqrt{dn})\}\cdot O(B)=\max\{O(B),O(Br_d\sqrt{dn})\}.\]
    Setting $r_d=C/d^2$ gives
    \[
    r_d\sqrt{dn} = C\frac{\sqrt n}{d^{3/2}}.
    \]
    Hence the fraction of bad indices satisfies
    \[
    \frac{\max\{O(B),O(Br_d\sqrt{dn})\}}{B\sqrt{n/d}}
    \;\lesssim\;
    \max\!\left\{\sqrt{\frac{d}{n}},\frac{C}{d}\right\}.
    \]
    For $d\le c n^{1/3}$ and sufficiently small $c$ and $C$, this is at most $1/(2d)$.
    In conclusion, the \textit{fraction} of the $\xi_j^{(m)}$'s, $m\in I_\mathsf{odd}$, within $r_d$ of an integer is at most $1/(2d)$.
    By a union bound, this means there exists an $n^*\in I_\mathsf{odd}$ such that \emph{all} roots $\xi_j^{(n^*)}$ are bounded away from $\Z$ by $C/d^{2}$, as desired.
\end{proof}

\noindent We are now ready to finish the main theorem.

\begin{proof}[Proof of \Cref{thm:high-disc-poly}]
    $K_d(t;m)$ factorizes as
    \[K_d(t;m) = \frac{(-2)^d}{d!}\prod_{j=1}^d(t-\xi_j^{(m)})\]
    Fix an integer $t$ and let $J=J(t)$ be the index of a closest root to $t$.
    Then
    \[|t-\xi_{J\pm k}|\geq \left(k-\tfrac12\right)\Delta_\mathrm{min}\,.\]
    and the worst (smallest) case is when $J$ is in the middle, so we get uniformly over $t\in \Z$,
    \begin{align*}
    |K_d(t;m)| &\geq \frac{2^d}{d!}\cdot 
    \mathop{\mathrm{dist}}\!\Big(\mathop{\mathrm{Roots}}\!\big(K_d(\,\cdot\,;m)\big)\,,\,\Z\Big)\cdot\Delta_\text{min}^{d-1}
    \cdot\left(\prod_{k=1}^{\lfloor(d-1)/2\rfloor}(k-\tfrac12)\right)^2\\
    &\gtrsim \frac{1}{d^3}\cdot\mathop{\mathrm{dist}}\!\Big(\mathop{\mathrm{Roots}}\!\big(K_d(\,\cdot\,;m)\big)\,,\,\Z\Big)\cdot\Delta_\text{min}^{d-1},
    \end{align*}
    where $\mathop{\mathrm{dist}}\!\Big(\mathop{\mathrm{Roots}}\!\big(K_d(\,\cdot\,;m)\big)\,,\,\Z\Big)$ denotes the quantity $\min_{s\in \Z,\,t\in \R:K_d(t;m)=0}|t-s|$.
    Setting $m=n^*$ from \Cref{cor:z-dist} and applying \Cref{fact:root-space} yields the final bound
    \[\min_{x\in\{\pm1\}^n}|e^{(n^*)}_d(x)|\geq\min_{t\in \Z}|K_d(t;n^*)|\gtrsim \frac{1}{d^{5}}\left(\frac{n}{d}\right)^{\frac{d-1}{2}}\,.\qedhere\]
\end{proof}

%% file: subspace.tex
\section{Hardness of learning all parameters and effective Hamiltonians}

We now turn our attention to our lower bounds for learning all the parameters (Theorem~\ref{thm:all-parameters-main}), and for effective Hamiltonian learning (Theorem~\ref{thm:effective-informal}), with algorithms in the sequential experiments model.
It turns out that both lower bounds can be proved using the same general lower bound framework.
We first reduce them both to instances of the Hamiltonian hypothesis selection problem.
\begin{definition}
    We define \emph{Hamiltonian Hypothesis Selection} problem as follows. Let $\mathcal{S}$ be a hypothesis set of $n$-qubit Hamiltonians, and let let $M \in S$ be an unknown Hamiltonian from $\mathcal{S}$.
    Given time evolution access to $M$, the goal of the learner is to output an $M$ with probability at least $0.1$.
\end{definition}
For this problem, we show:
\begin{lemma}
\label{lem:subspace-holevo}
    Let $\mathcal{S}$ be a family of $n$-qubit Hamiltonians.
     Any algorithm in the sequential experiments model (as defined in \Cref{def:access}) for the hypothesis selection problem for $\mathcal{S}$ requires at least $m\ge \Omega(\log(|\mathcal{S}|)/ n)$ rounds of interactions with the unknown Hamiltonian.
\end{lemma}
\noindent
The key technical ingredient we need for proving~\Cref{lem:subspace-holevo} is the following geometric lemma, which states that we can control the geometry of the quantum information that any algorithm in the sequential experiments model can obtain:
\begin{lemma}
\label{lem:subspace-counting}
    For any fixed unitaries $U_j$, let $S_m (U_1, \ldots, U_m)$ denote the following subset of $(\mathbb{C}^2)^{\otimes (n+r)}$:
    \[
        S_m(U_1,U_2,\cdots U_m) := \left\{ \ket{\varphi} = \prod_{j = 1}^m \left(V_j \otimes I \right) U_j \ket{\psi}: V_j \text{ are unitaries on the first $n$ qubits}, \forall j\in [m]
        \right\} \; .
    \]
    Then, $S_m (U_1, \ldots, U_m)$ lies in a subspace of dimension $2^{2mn}$.
\end{lemma}
\begin{proof}
    Let $N = 2^n$. We prove by induction on $m$. This is definitely true when $m = 0$. Suppose $S_{m-1}\subseteq \mathrm{span}(\ket{\alpha_1}\cdots \ket{\alpha_{L}})$, where $L \le 2^{2(m-1)n}$.

    Write $U_m\ket{\alpha_i} = \sum_{j=1}^{N} c_{i,j} \ket{j}\ket{\beta_{i,j}}$ for $i\in [L]$.
    Then
    \[
        (V_m \otimes I)U_m \ket{\alpha_i} = \sum_{j,k\in [N]} c_{i,j} (V_m)_{j,k} \ket{k} \ket{\beta_{i,j}} ] \; , 
    \]
    and so we have $S_m \subseteq \mathrm{span}(\{\ket{k}\ket{\beta_{i,j}}\}_{i\in [L],j,k\in [N]})$, which is a subspace of dimension at most $N^2L \le 2^{2mn}$, as we desired.
\end{proof}
\noindent With this, we can now prove~\Cref{lem:subspace-holevo}:
\begin{proof}[Proof of~\Cref{lem:subspace-holevo}]
    Let $X$ be a uniformly random Hamiltonian from the hypothesis set $\mathcal{S}$. 
    Recall that an algorithm in the sequential experiments model makes $R$ total rounds of measurements, for some parameter $R$.
    Let $\ell \leq R$ be some round of interaction.
    Condition on the measurement outcomes observed by the algorithms in rounds $1, \ldots, \ell - 1$.
    Recall that $m^{(\ell)}$ is the total number of rounds of interaction that the algorithm will make in the $\ell$-th round, and that by definition, $m \geq \sum_{\ell = 1}^R m^{(\ell)}$.
    We will argue that the algorithm can extract at most $2m^{(\ell)} n$ bits of information about $X$ in the $\ell$-th round.
    Let $Y^{(\ell)}$ be the measurement outcome of the algorithm in the $\ell$-th round, for $\ell = 1, \ldots, R$. Let $\ket{\varphi_X^{(\ell)}}$ be the (random) state that the algorithm measures in the $\ell$-th round of measurements, if the underlying Hamiltonian is $X$:
    \[
        \ket{\varphi_X^{(\ell)}} \coloneqq \prod_{j = 1}^{m^{(\ell)}} \left( e^{-iX t_j} \otimes I \right) U_j \ket{\psi^{(\ell)}} \, ,    
    \]
    Then, by~\Cref{lem:subspace-counting}, we know that $\ket{\varphi_X^{(\ell)}}$ lies in a fixed subspace of dimension $2^{2 m n}$.

    Then by Holevo's inequality \cite{nielsen2010quantum}, we know the mutual information between $X$ and $Y^{(\ell)}$ conditioned on $Y^{(0)},\dots ,Y^{(\ell-1)}$ is bounded:
    \[
        I(X:Y^{(\ell)} \mid Y^{(0)},\dots ,Y^{(\ell-1)}) \le S\left(\sum_{X\in \mathcal{S}} p_X^{(\ell)} \ketbra{\varphi_X^{(\ell)}}{\varphi_X^{(\ell)}}\right) \le \log(2^{2m^{(\ell)} n}) = 2m^{(\ell)}n,
    \]where $p_X^{(\ell)}$ is the conditional probability: $P(X\mid Y^{(0)},\dots ,Y^{(\ell-1)})$.
    
    By definition, we have $I(X:Y\mid Y^{(0)},\dots ,Y^{(\ell-1)}) = H(X\mid Y^{(0)},\dots ,Y^{(\ell-1)}) - H(X\mid Y)$, and so by telescoping, we obtain:
    \begin{align*}
        &  \qquad H(X) - H(X\mid Y^{(0)},\dots ,Y^{(R)}) \\ &= \sum_{\ell = 1}^R \E_{Y^{(0)},\dots ,Y^{(\ell-1)}}[H(X\mid Y^{(0)},\dots ,Y^{(\ell-1)}) - H(X\mid Y^{(0)},\dots ,Y^{(\ell)})]\\
        &= \E_{Z^{(0)},\dots,Z^{(R)}} \left[\sum_{\ell =1}^R H(X\mid Z^{(0)},\dots ,Z^{(\ell-1)}) - H(X\mid Z^{(0)},\dots ,Z^{(\ell-1)}, Y^{(\ell)})\right]\\
        &\le \E\left[\sum_{\ell=1}^R 2m^{(\ell)}n\right] \le 2mn.
    \end{align*}
    Hence, we have that $H(X\mid Y^{(0)},\dots ,Y^{(\ell)}) \ge \log(|\mathcal{S}|) - 2mn$. 
    
    On the other hand, by Fano's inequality \cite{cover2006elements}, since we assume that the learner correctly recovers $X$ with probability at least $0.1$, we obtain that:
    \[
        H(X\mid Y) \le 1 + 0.9 \log(|\mathcal{S}|) \; .
    \]
    \noindent
    Putting these bounds together, we obtain $\log(|\mathcal{S}|) - 2mn \le 1 + 0.9 \log(|\mathcal{S}|)$, and so $m\ge \Omega(\log(|\mathcal{S}|)/n)$, as claimed.
\end{proof}

\noindent
In the following, we apply this lemma to those learning problems containing a large $\Omega(1)$-packing set to obtain strong lower bounds.

\subsection{Lower bounds for learning all parameters: proof of~\Cref{thm:all-parameters-main}}

We first demonstrate how to instantiate this bound to obtain~\Cref{thm:all-parameters-main}.

\begin{proof}[Proof of~\Cref{thm:all-parameters-main}]
    Consider the hypothesis selection problem where hypothesis set $\mathcal{S}$ is defined to be: 
    \[\mathcal{S} \coloneqq \left\{ H=\sum_{P\in \mathcal{P}}\alpha_P P:\alpha_p\in\{\pm 1\} \right\} \; .\]
    An algorithm which learns all of the parameters up to accuracy $\eps = 0.99$ also solves this hypothesis selection problem with non-trivial probability. Hence by Lemma~\ref{lem:subspace-holevo}, the number of interactions is at least $m\ge \Omega(\log(|\mathcal{S}|)/n) = \Omega(|\mathcal{P}|/n)$.
\end{proof}
\noindent

\subsection{Lower bounds for learning effective Hamiltonians}

In this section, we show our lower bound against learning effective non-parametric approximations to the target Hamiltonian.
Formally, we show:

\begin{theorem}
\label{thm:effective-main}
    Let $\delta \in (0,1/2)$, and let $n$ be sufficiently large.
    Let $t^* > 0$ be some time parameter, and let $k > 0$.
 Then, any algorithm that solves the effective Hamiltonian learning problem for $k$-local, $n$-qubit Hamiltonians at time $t^*$ with $m \leq O(F(n,k)/n^{1+1/\log(1/\delta)})$ queries must pay error at least $\epsilon = \Omega (\delta \min (1, t^*))$, 
    where
    \[
    F(n,k) \coloneqq \max_{t\in(0,\frac12 - \frac{1}{\log n})}  \min\left(n^{(1/2-t)k},\exp(cn^{2t})\right) \; ,
    \]
    for some universal constant $c>0$.
\end{theorem}
\noindent
We note that achieving error $\eps=O (\min (1, t^*))$ is trivial under the assumption that the spectral norm of the Hamiltonian is at most $1$, as this error guarantee is achieved by outputting the identity matrix.

We pause here to interpret this theorem, and specifically, the quantity $F(n,k)$.
When $k$ is a small constant, note that $F(n, k) = n^{k/2 - o(1)}$.
More generally, by letting $t = \tfrac{1}{4} + \tfrac{1}{4} \tfrac{\log k}{\log n}$, one can show that $F(n, k) \geq \max \left( (n / k)^{\Omega (k)}, 2^{\Omega (k)} \right)$ for all $k \leq n / 2$---in particular, note this bound is always polynomial in the number of parameters in the system.
Thus our result shows that achieving any non-trivial error for this problem always requires a number of queries $m$ which is polynomial in the number of total parameters in the system, as long as $k \geq 3$.

Hence our lower bound states that no interesting learning guarantees are possible without $n^{O(k)}$-many queries to the time evolution of the unknown Hamiltonian.
On the other hand, we note that with $n^{\Omega (k)}$ queries, the algorithm of \textit{e.g.}~\cite{bakshi2024structure} already allows us to not only learn an effective Hamiltonian, but also learn the parameters of the underlying Hamiltonian.
Thus, at a high level, our result yields another qualitative ``all-or-nothing'' phenomenon for even this non-parametric notion of Hamiltonian learning: it states that up to the constant in the exponent of $k$, one cannot hope for more efficient algorithms by relaxing the problem from learning the parameters to this weaker guarantee.

\subsection{An $\ell_\infty$ packing of degree-$k$ polynomials on the hypercube}

In this subsection we show that there exists a large packing set of $\ell_\infty$ bounded degree-$k$ polynomials on the hypercube.

\begin{lemma}
\label{lem:packing-ell-infty}
    For all $\delta \in (0,1/2)$, there exists a set $S$ of polynomials with degree at most $k$ of size $K$ such that for any $f \in S$, $\norm{f}_\infty \le 1$ and for all $f, g \in S$ so that $f \neq g$, we have that $\norm{f-g}_\infty \ge \delta$, where $K$ is given by:
    \[
        K = \exp\left(\Omega \left( \frac{F(n, k)}{n^{1/\log(1/\delta)}} \right)\right) \; ,
    \]
\end{lemma}
\noindent 
For any $k$, let $M_k(v) = \max_{x\in \{\pm 1\}^n}\langle v,x^{\otimes k} \rangle$, for $v \in \mathbb{R}^{n^k}$.
The proof of Lemma~\ref{lem:packing-ell-infty} relies on the following key lemma:
\begin{lemma}
\label{lem:packing-gaussian}
    For $g\sim N(0,I_{n^k})$, 
    \[
        \Pr[M_k(g) \le n^{k/2}]\le \exp(-F(n,k)).  
    \]
    where $c>0$ is some universal constant.
\end{lemma}
\begin{proof}
    Let $t\in (0,\frac12-\frac{1}{\log n})$ be any constant.
    First we show that there exists a set $V\subseteq \{\pm 1\}^n$ with $|V| \ge \Omega(\exp(n^{2t}))$ such that for any $x\neq x'\in V$, we have that $|x\cdot x'|\ge n^{1/2+t}$. 
    
    We show this by a volume argument. 
    Suppose $V$ is a set satisfying for any $x\neq x'\in V$.
    Then $|\langle x, x'\rangle| \ge n^{1/2+t}$.
    We will show that unless $|V| \le \exp(\Omega(n^{2t}))$, we can add a point to $V$ while maintaining this property, which will prove the claim.
    Consider sampling a random point $y\sim_U \{\pm 1\}^n$. 
    For any $x$, we have that
    \[
    \Pr \left[ |\langle x, y\rangle| > n^{1/2+t} \right] \leq \exp(-\Omega(n^{2t})) \; .
    \]
    Hence by union bound we have
    \[
        \Pr[|\langle x, y\rangle| \le n^{1/2+t}, \forall x\in |V|] \ge 1 - |V| \cdot \exp(-\Omega(n^{2t})) > 0 \; ,
    \]
    unless $|V| \le \exp(\Omega(n^{2t}))$, so there exists a $y$ we can add to $V$, as claimed.


    Let $V$ be such a set.
    Define the Gaussian process $X_x = \langle g,x^{\otimes k} \rangle$ for $x\in V$.
    Then, we have that $\E[X_x^2] = n^{k}$ and $\E[X_xX_y] \le n^{(1/2+t)k}$ for $x\neq y$. Hence by Slepian's lemma \cite{vershynin2018high},
    \[
        \Pr[\max_{x\in C} X_x \le n^{k/2}] \le \Pr_{z \sim N(0,A)}\left[\max_{i\in C} z_i \le n^{k/2}\right],
    \]where $A = (n^{k}-n^{(1/2+t)k})I + n^{(1/2+t)k}\mathbf{1}\mathbf{1}^T$.

    Let $h,g_i$ for $i\in V$ be independent univariate Gaussians, where $\E[h^2] = n^{(1/2+t)k}$ and $\E[g_i^2] = n^k - n^{(1/2+t)k}$.
   Then $z$ has same distribution with $(g_i+h)_{i\in V}$.
   Note that since $t\le \frac12 - \frac{1}{\log n}$, we have that $\E[g_i^2] \geq n^{k} / 2$.
    Then by union bound and independence of $g_i$, $i\in V$:
    \begin{align*}
        \Pr[\max_{i\in C}g_i + h\le n^{k/2}] &\le \Pr[h\le -n^{k/2}] + \Pr[\max_{i\in V}g_i \le 2n^{k/2}] \\
        &\le \exp(-\Omega(n^{(1/2-t)k})) + \exp(-\Omega(|V|)),
    \end{align*}
    as we desired.
\end{proof}
\begin{proof}[Proof of~\Cref{lem:packing-ell-infty}] We show the existence of this large packing set by a probabilistic argument. 
Let $\zeta>0$ be a parameter to be specified later.
We take $L$ independent $g_1, \ldots, g_L \sim N(0,I)$, and delete those $g_i$ which violate the condition that $|M_k(g)|\le \zeta$ or any pair of $g_i \neq g_j$ that violate the condition that $|M_k(g_i-g_j)| \ge \delta \zeta $. After this deleting step, define the set of polynomials $f_i = \langle g_i, x^{\otimes k}\rangle / \zeta$ for all remaining $g_i$.
Then, this set of functions $f_i$ satisfies the conditions of the packing.
It thus suffices to show that there exists such a $\zeta$ so that the expected number of $g_i$ which remain after this deletion step is large.

Let $p (\zeta) = \Pr_{g\sim N(0,I)}[|M_k(g)| \le \zeta]$. Notice that for $i \neq j$, we have that $g_i-g_j$ has same distribution as $N(0,2I)$, so for any pair $i \neq j$, the probability that $|M_k(g_i-g_j)| \ge \delta\zeta$ is violated is $p(\delta\zeta/\sqrt{2}) \le p(\delta \zeta)$. 
Hence the expected size of the resulting set after deleting all violating elements is at least $L p(\zeta) - L^2 p(\delta\zeta)$.
Let $L = p(\zeta) / (2p(\delta\zeta))$, so that the expected size becomes $p^2(\zeta)/4p(\delta\zeta)$.
Hence we only need to show the existence of $\zeta$ such that $p^2(\zeta)/p(\delta \zeta)$ is at least $K$.

Suppose $p^2(\zeta) \le K p(\delta\zeta)$, for all $\zeta$. Then 
\[(p(\zeta)/K)^2 \le p(\delta \zeta)/K.\]
Iteratively applying this $\lceil\log n/\log (1/\delta)\rceil+2$ times, we get
\[(p(\zeta)/K)^{8n^{1/\log(1/\delta)}} \le p(\zeta/(2n))/K.\]
If we let $\zeta=2n^{(k+1)/2}$, then we have that $p(\zeta) \ge 1/2$ by simple union bound, but on the other hand, we have that $p(\zeta/(2n)) \le p(n^{k/2}) \le (2K)^{-10n^{1/\log(1/\delta)}}$ by Lemma~\ref{lem:packing-gaussian}.
Putting these inequalities together, we obtain,
\[
    (1/2K)^{8n^{1/\log(1/\delta)}} \le (2K)^{-10n^{1/\log(1/\delta)}}
\] which is a contradiction.
This finishes the proof.
%
\end{proof}

\begin{proof}[Proof of \Cref{thm:effective-main}]
    Suppose the algorithm learns the effective Hamiltonian up to accuracy $\eps = c \delta \cdot \min(t^*,1)$ for some $c$ sufficiently small. 
    Then it can be used for the hypothesis selection problem where the hypothesis set $\mathcal{S}$ is a set satisfying the condition that for any two $H\neq H'\in \mathcal{S}$, we have that $\norm{\mathcal{U}(e^{-iHt^*})-\mathcal{U}(e^{-iH't^*})}_\diamond \ge 3\eps$.
    
    Construct the hypothesis set as follows. 
    Let $\mathcal{F}$ be the set of polynomials on the hypercube in Lemma~\ref{lem:packing-ell-infty}. 
    For each function $f$ we define its corresponding Hamiltonian $M^{(f)}$ where it is only supported on $Z_S$ and its vector $\vec{\alpha}$ is defined as $\alpha_{Z_S} = \widehat{f}(S)$ for all $|S|\le k$. Then $\langle x| M^{(f)}| x \rangle = f(x)$ for all $x\in \{\pm 1\}^n$. Since $M^{(f)}$ is diagonal under computational basis and $\norm{f}_\infty \le 1$, we have $\norm{M^{(f)}}_\infty \le 1$.

    Notice that $\E_{x\in \{\pm 1\}^n} f(x) \in [-1,1]$ for $f\in \mathcal{F}$. For $p\in [-1,1]$, we define $\mathcal{F}_p = \{f:\E_{x\in \{\pm 1\}^n} f(x) \in [p,p+0.1\delta]\}$. By the pigeonhole principle, there exist $q$ such that $|\mathcal{F}_q| \ge 0.1 \cdot \delta |\mathcal{F}|$.

    When $t^*>1$, so the desired accuracy is $\eps = c \delta$, we can reduce this to the case where $t^* = 1$ by considering the hypothesis set $\mathcal{S} = \{M^{(f)}/t^*\}_{f\in \mathcal{F}_q}$. When $t^*\le 1$, so the desired accuracy is $\eps=c\delta t^*$, we let the hypothesis set $\mathcal{S}$ be $\{M^{(f)}\}_{f\in \mathcal{F}_q}$. Then for any $f\neq g\in \mathcal{F}$,
    \begin{align*}
        \norm{\mathcal{U}(e^{-iM^{(f)}t^*})-\mathcal{U}(e^{-iM^{(g)}t^*})}_\diamond &\ge \dph(e^{-iM^{(f)}t^*}, e^{-iM^{(g)}t^*}) \\
        &= \min_{\theta\in[-\pi,\pi]} \norm{e^{-iM^{(f)}t}\cdot e^{i\theta}-e^{-iM^{(g)}t^*}} \\
        &= \min_{\theta\in [-\pi,\pi]} \max_{x\in\{\pm 1\}^n} |e^{-if(x)t^*+i\theta} - e^{-ig(x)t^*}| \\
        &= \min_{\theta\in [-\pi,\pi]} \max_{x\in\{\pm 1\}^n} |1 - e^{i(f(x)t^*-g(x)t^*-\theta)}| \\
        &\ge \min_{\theta\in [-\pi,\pi]} \max_{x\in\{\pm 1\}^n} 0.2\cdot |f(x)t^*-g(x)t^*-\theta| \\
        &= 0.2\min_{\theta\in [-\pi,\pi]} \norm{t^*(f-g)-\theta}_\infty,
    \end{align*}
    where the last inequality is by that $|f(x)t^*-g(x)t^*-\theta| \le 2+\pi (< 2\pi)$.

    Now for $|\theta| \le 0.3\delta t^*$, we have $ \norm{t^*(f-g)-\theta}_\infty \ge  \norm{t^*(f-g)}_\infty-\theta \ge \delta t^* - \theta \ge 0.7 \delta t^*$. For $|\theta| > 0.3 \delta t^*$, we have $\norm{t^*(f-g)-\theta}_\infty \ge |\E_{x\in \{\pm 1\}^n}[t^*(f(x)-g(x))-\theta]| \ge \theta - 0.1\delta t^* = 0.2\delta t^*$. Hence in either case we have the diamond distance is at least $3\eps$.
    

    By plugging in Lemma~\ref{lem:packing-ell-infty} and Lemma~\ref{lem:subspace-holevo}, the number of interactions $m$ for the algorithm is at least 
    \[\Omega(\log(|\mathcal{F}_q|)/n) \ge \Omega(F(n,k)/n^{1+1/\log(1/\delta)}) \; .\qedhere\]
\end{proof}

%% file: bibliography.bib
@book{cover2006elements,
  title     = {Elements of Information Theory},
  author    = {Cover, Thomas M. and Thomas, Joy A.},
  edition   = {2},
  year      = {2006},
  publisher = {Wiley-Interscience},
  isbn      = {978-0-471-24195-9}
}

@book{nielsen2010quantum,
  title={Quantum computation and quantum information},
  author={Nielsen, Michael A and Chuang, Isaac L},
  year={2010},
  publisher={Cambridge university press}
}

@book{vershynin2018high,
  title={High-dimensional probability: An introduction with applications in data science},
  author={Vershynin, Roman},
  volume={47},
  year={2018},
  publisher={Cambridge university press}
}

@article{aharonov2022quantum,
  title={Quantum algorithmic measurement},
  author={Aharonov, Dorit and Cotler, Jordan and Qi, Xiao-Liang},
  journal={Nature communications},
  volume={13},
  number={1},
  pages={887},
  year={2022},
  publisher={Nature Publishing Group UK London}
}

@inproceedings{bubeck2020entanglement,
  title={Entanglement is necessary for optimal quantum property testing},
  author={Bubeck, Sebastien and Chen, Sitan and Li, Jerry},
  booktitle={2020 IEEE 61st Annual Symposium on Foundations of Computer Science (FOCS)},
  pages={692--703},
  year={2020},
  organization={IEEE}
}

@book{duhamel1860elements,
  title={{\'E}l{\'e}ments de calcul infinit{\'e}simal},
  author={Duhamel, Jean Marie Constant},
  volume={1},
  year={1860},
  publisher={Mallet-Bachelier}
}

@book{Koekoek2010,
  title = {Hypergeometric Orthogonal Polynomials and Their q-Analogues},
  ISBN = {9783642050145},
  ISSN = {1439-7382},
  url = {http://dx.doi.org/10.1007/978-3-642-05014-5},
  DOI = {10.1007/978-3-642-05014-5},
  journal = {Springer Monographs in Mathematics},
  publisher = {Springer Berlin Heidelberg},
  author = {Koekoek,  Roelof and Lesky,  Peter A. and Swarttouw,  René F.},
  year = {2010}
}

@book{Kato1995,
  author    = {Tosio Kato},
  title     = {Perturbation Theory for Linear Operators},
  publisher = {Springer},
  address   = {Berlin, Heidelberg},
  series    = {Classics in Mathematics},
  year      = {1995},
  isbn      = {978-3-540-58661-6}
}

@BOOK{Gautschi2004-il,
  title     = "Orthogonal Polynomials",
  author    = "Gautschi, Walter",
  publisher = "Oxford University Press",
  series    = "Numerical Mathematics and Scientific Computation",
  month     =  apr,
  year      =  2004,
  address   = "London, England",
  language  = "en"
}

@book{watrous2018theory,
  title={The theory of quantum information},
  author={Watrous, John},
  year={2018},
  publisher={Cambridge university press}
}

@article{stilck2024efficient,
  title={Efficient and robust estimation of many-qubit Hamiltonians},
  author={Stilck Fran{\c{c}}a, Daniel and Markovich, Liubov A and Dobrovitski, Viatcheslav V and Werner, Albert H and Borregaard, Johannes},
  journal={Nature Communications},
  volume={15},
  number={1},
  pages={311},
  year={2024},
  publisher={Nature Publishing Group UK London}
}

@inproceedings{haah2023query,
  title={Query-optimal estimation of unitary channels in diamond distance},
  author={Haah, Jeongwan and Kothari, Robin and O’Donnell, Ryan and Tang, Ewin},
  booktitle={2023 IEEE 64th Annual Symposium on Foundations of Computer Science (FOCS)},
  pages={363--390},
  year={2023},
  organization={IEEE}
}

@book{SymmEigs,
author = {Parlett, Beresford N.},
title = {The Symmetric Eigenvalue Problem},
publisher = {Society for Industrial and Applied Mathematics},
year = {1998},
doi = {10.1137/1.9781611971163},
address = {},
edition   = {},
URL = {https://epubs.siam.org/doi/abs/10.1137/1.9781611971163},
eprint = {https://epubs.siam.org/doi/pdf/10.1137/1.9781611971163}
}

@article{KRASIKOV2009121,
title = {On zeros of discrete orthogonal polynomials},
journal = {Journal of Approximation Theory},
volume = {156},
number = {2},
pages = {121-141},
year = {2009},
issn = {0021-9045},
doi = {https://doi.org/10.1016/j.jat.2008.04.015},
url = {https://www.sciencedirect.com/science/article/pii/S0021904508001445},
author = {Ilia Krasikov and Alexander Zarkh},
}

@article{chen2025efficient,
  title={Efficient Pauli channel estimation with logarithmic quantum memory},
  author={Chen, Sitan and Gong, Weiyuan},
  journal={PRX Quantum},
  volume={6},
  number={2},
  pages={020323},
  year={2025},
  publisher={APS}
}

@article{chen2024tight,
  title={Tight bounds on Pauli channel learning without entanglement},
  author={Chen, Senrui and Oh, Changhun and Zhou, Sisi and Huang, Hsin-Yuan and Jiang, Liang},
  journal={Physical Review Letters},
  volume={132},
  number={18},
  pages={180805},
  year={2024},
  publisher={APS}
}

@article{bluhm2024hamiltonian,
  title={Hamiltonian property testing},
  author={Bluhm, Andreas and Caro, Matthias C and Oufkir, Aadil},
  journal={arXiv preprint arXiv:2403.02968},
  year={2024}
}

@article{francca2025learning,
  title={Learning and certification of local time-dependent quantum dynamics and noise},
  author={Fran{\c{c}}a, Daniel Stilck and M{\"o}bus, Tim and Rouz{\'e}, Cambyse and Werner, Albert H},
  journal={arXiv preprint arXiv:2510.08500},
  year={2025}
}

@article{sinha2025improved,
  title={Improved Hamiltonian learning and sparsity testing through Bell sampling},
  author={Sinha, Savar D and Tong, Yu},
  journal={arXiv preprint arXiv:2509.07937},
  year={2025}
}

@book{o2014analysis,
  title={Analysis of boolean functions},
  author={O'Donnell, Ryan},
  year={2014},
  publisher={Cambridge University Press}
}

@misc{FOCS2024QuantumLearningOpenQuestions,
  author       = {Chen, Sitan and Cotler, Jordan and Huang, Hsin-Yuan and Li, Jerry and Tang, Ewin},
  title        = {Open Questions Session, Recent Advances in Quantum Learning (FOCS 2024 Workshop)},
  howpublished = {Workshop presentation},
  year         = {2024},
  month        = oct,
  day          = {27},
  note         = {FOCS 2024, Chicago. Includes list of open problems in quantum learning.},
  url          = {https://jerryzli.github.io/focs24-workshop.html},
  accessed     = {2025-09-12}
}

@article{verdon2020quantum,
  title={Quantum Hamiltonian-Based Models and the Variational Quantum Thermalizer Algorithm},
  author={Verdon, Guillaume and Marks, Jacob and Nanda, Sasha and Leichenauer, Stefan and Hidary, Jack},
  journal={Bulletin of the American Physical Society},
  volume={65},
  year={2020},
  publisher={APS}
}

@article{wang2017experimental,
  title={Experimental quantum Hamiltonian learning},
  author={Wang, Jianwei and Paesani, Stefano and Santagati, Raffaele and Knauer, Sebastian and Gentile, Antonio A and Wiebe, Nathan and Petruzzella, Maurangelo and O’brien, Jeremy L and Rarity, John G and Laing, Anthony and others},
  journal={Nature Physics},
  volume={13},
  number={6},
  pages={551--555},
  year={2017},
  publisher={Nature Publishing Group UK London}
}

@article{wiebe2014hamiltonian,
  title={Hamiltonian learning and certification using quantum resources},
  author={Wiebe, Nathan and Granade, Christopher and Ferrie, Christopher and Cory, David G},
  journal={Physical review letters},
  volume={112},
  number={19},
  pages={190501},
  year={2014},
  publisher={APS}
}

@article{wiebe2014quantum,
  title={Quantum Hamiltonian learning using imperfect quantum resources},
  author={Wiebe, Nathan and Granade, Christopher and Ferrie, Christopher and Cory, David},
  journal={Physical Review A},
  volume={89},
  number={4},
  pages={042314},
  year={2014},
  publisher={APS}
}

@article{sundaresan2020reducing,
  title={Reducing unitary and spectator errors in cross resonance with optimized rotary echoes},
  author={Sundaresan, Neereja and Lauer, Isaac and Pritchett, Emily and Magesan, Easwar and Jurcevic, Petar and Gambetta, Jay M},
  journal={PRX Quantum},
  volume={1},
  number={2},
  pages={020318},
  year={2020},
  publisher={APS}
}

@article{sheldon2016procedure,
  title={Procedure for systematically tuning up cross-talk in the cross-resonance gate},
  author={Sheldon, Sarah and Magesan, Easwar and Chow, Jerry M and Gambetta, Jay M},
  journal={Physical Review A},
  volume={93},
  number={6},
  pages={060302},
  year={2016},
  publisher={APS}
}

@article{shulman2014suppressing,
  title={Suppressing qubit dephasing using real-time Hamiltonian estimation},
  author={Shulman, Michael D and Harvey, Shannon P and Nichol, John M and Bartlett, Stephen D and Doherty, Andrew C and Umansky, Vladimir and Yacoby, Amir},
  journal={Nature communications},
  volume={5},
  number={1},
  pages={5156},
  year={2014},
  publisher={Nature Publishing Group UK London}
}

@article{innocenti2020supervised,
  title={Supervised learning of time-independent Hamiltonians for gate design},
  author={Innocenti, Luca and Banchi, Leonardo and Ferraro, Alessandro and Bose, Sougato and Paternostro, Mauro},
  journal={New Journal of Physics},
  volume={22},
  number={6},
  pages={065001},
  year={2020},
  publisher={IOP Publishing}
}

@article{boulant2003robust,
  title={Robust method for estimating the Lindblad operators of a dissipative quantum process from measurements of the density operator at multiple time points},
  author={Boulant, Nicolas and Havel, Timothy F and Pravia, Marco A and Cory, David G},
  journal={Physical Review A},
  volume={67},
  number={4},
  pages={042322},
  year={2003},
  publisher={APS}
}

@article{caves1981quantum,
  title={Quantum-mechanical noise in an interferometer},
  author={Caves, Carlton M},
  journal={Physical Review D},
  volume={23},
  number={8},
  pages={1693},
  year={1981},
  publisher={APS}
}

@article{wineland1992spin,
  title={Spin squeezing and reduced quantum noise in spectroscopy},
  author={Wineland, David J and Bollinger, John J and Itano, Wayne M and Moore, FL and Heinzen, Daniel J},
  journal={Physical Review A},
  volume={46},
  number={11},
  pages={R6797},
  year={1992},
  publisher={APS}
}

@article{holland1993interferometric,
  title={Interferometric detection of optical phase shifts at the Heisenberg limit},
  author={Holland, Murray J and Burnett, Keith},
  journal={Physical review letters},
  volume={71},
  number={9},
  pages={1355},
  year={1993},
  publisher={APS}
}

@article{mckenzie2002experimental,
  title={Experimental demonstration of a squeezing-enhanced power-recycled Michelson interferometer for gravitational wave detection},
  author={McKenzie, Kirk and Shaddock, Daniel A and McClelland, David E and Buchler, Ben C and Lam, Ping Koy},
  journal={Physical review letters},
  volume={88},
  number={23},
  pages={231102},
  year={2002},
  publisher={APS}
}

@article{leibfried2004toward,
  title={Toward Heisenberg-limited spectroscopy with multiparticle entangled states},
  author={Leibfried, Dietrich and Barrett, Murray D and Schaetz, Tobias and Britton, Joseph and Chiaverini, John and Itano, Wayne M and Jost, John D and Langer, Christopher and Wineland, David J},
  journal={Science},
  volume={304},
  number={5676},
  pages={1476--1478},
  year={2004},
  publisher={American Association for the Advancement of Science}
}

@article{bollinger1996optimal,
  title={Optimal frequency measurements with maximally correlated states},
  author={Bollinger, John J and Itano, Wayne M and Wineland, David J and Heinzen, Daniel J},
  journal={Physical Review A},
  volume={54},
  number={6},
  pages={R4649},
  year={1996},
  publisher={APS}
}

@article{valencia2004distant,
  title={Distant clock synchronization using entangled photon pairs},
  author={Valencia, Alejandra and Scarcelli, Giuliano and Shih, Yanhua},
  journal={Applied Physics Letters},
  volume={85},
  number={13},
  pages={2655--2657},
  year={2004},
  publisher={AIP Publishing}
}

@article{narayanan2024improved,
  title={Improved algorithms for learning quantum Hamiltonians, via flat polynomials},
  author={Narayanan, Shyam},
  journal={arXiv preprint arXiv:2407.04540},
  year={2024}
}

@article{garcia2024estimation,
  title={Estimation of Hamiltonian parameters from thermal states},
  author={Garc{\'\i}a-Pintos, Luis Pedro and Bharti, Kishor and Bringewatt, Jacob and Dehghani, Hossein and Ehrenberg, Adam and Yunger Halpern, Nicole and Gorshkov, Alexey V},
  journal={Physical Review Letters},
  volume={133},
  number={4},
  pages={040802},
  year={2024},
  publisher={APS}
}

@article{fawzi2024certified,
  title={Certified algorithms for equilibrium states of local quantum Hamiltonians},
  author={Fawzi, Hamza and Fawzi, Omar and Scalet, Samuel O},
  journal={Nature Communications},
  volume={15},
  number={1},
  pages={7394},
  year={2024},
  publisher={Nature Publishing Group UK London}
}

@inproceedings{bakshi2024learning,
  title={Learning quantum Hamiltonians at any temperature in polynomial time},
  author={Bakshi, Ainesh and Liu, Allen and Moitra, Ankur and Tang, Ewin},
  booktitle={Proceedings of the 56th Annual ACM Symposium on Theory of Computing},
  pages={1470--1477},
  year={2024}
}

@article{anshu2021efficient,
  title={Efficient learning of commuting Hamiltonians on lattices},
  author={Anshu, Anurag and Arunachalam, Srinivasan and Kuwahara, Tomotaka and Soleimanifar, Mehdi},
  journal={Electronic notes},
  volume={25},
  year={2021}
}

@article{dutkiewicz2024advantage,
  title={The advantage of quantum control in many-body Hamiltonian learning},
  author={Dutkiewicz, Alicja and O'Brien, Thomas E and Schuster, Thomas},
  journal={Quantum},
  volume={8},
  pages={1537},
  year={2024},
  publisher={Verein zur F{\"o}rderung des Open Access Publizierens in den Quantenwissenschaften}
}

@article{ma2024learning,
  title={Learning $ k $-body Hamiltonians via compressed sensing},
  author={Ma, Muzhou and Flammia, Steven T and Preskill, John and Tong, Yu},
  journal={arXiv preprint arXiv:2410.18928},
  year={2024}
}

@article{hu2025ansatz,
  title={Ansatz-free Hamiltonian learning with Heisenberg-limited scaling},
  author={Hu, Hong-Ye and Ma, Muzhou and Gong, Weiyuan and Ye, Qi and Tong, Yu and Flammia, Steven T and Yelin, Susanne F},
  journal={arXiv preprint arXiv:2502.11900},
  year={2025}
}

@inproceedings{zhao2025learning,
  title={Learning the structure of any Hamiltonian from minimal assumptions},
  author={Zhao, Andrew},
  booktitle={Proceedings of the 57th Annual ACM Symposium on Theory of Computing},
  pages={1201--1211},
  year={2025}
}

@inproceedings{arunachalam2025testing,
  title={Testing and learning structured quantum Hamiltonians},
  author={Arunachalam, Srinivasan and Dutt, Arkopal and Escudero Guti{\'e}rrez, Francisco},
  booktitle={Proceedings of the 57th Annual ACM Symposium on Theory of Computing},
  pages={1263--1270},
  year={2025}
}

@inproceedings{bakshi2024structure,
  title={Structure learning of Hamiltonians from real-time evolution},
  author={Bakshi, Ainesh and Liu, Allen and Moitra, Ankur and Tang, Ewin},
  booktitle={2024 IEEE 65th Annual Symposium on Foundations of Computer Science (FOCS)},
  pages={1037--1050},
  year={2024},
  organization={IEEE}
}

@article{ni2024quantum,
  title={Quantum hamiltonian learning for the fermi-hubbard model},
  author={Ni, Hongkang and Li, Haoya and Ying, Lexing},
  journal={Acta Applicandae Mathematicae},
  volume={191},
  number={1},
  pages={2},
  year={2024},
  publisher={Springer}
}

@article{mirani2024learning,
  title={Learning interacting fermionic Hamiltonians at the Heisenberg limit},
  author={Mirani, Arjun and Hayden, Patrick},
  journal={Physical Review A},
  volume={110},
  number={6},
  pages={062421},
  year={2024},
  publisher={APS}
}

@article{li2024heisenberg,
  title={Heisenberg-limited Hamiltonian learning for interacting bosons},
  author={Li, Haoya and Tong, Yu and Gefen, Tuvia and Ni, Hongkang and Ying, Lexing},
  journal={npj Quantum Information},
  volume={10},
  number={1},
  pages={83},
  year={2024},
  publisher={Nature Publishing Group UK London}
}

@article{huang2023learning,
  title={Learning many-body Hamiltonians with Heisenberg-limited scaling},
  author={Huang, Hsin-Yuan and Tong, Yu and Fang, Di and Su, Yuan},
  journal={Physical Review Letters},
  volume={130},
  number={20},
  pages={200403},
  year={2023},
  publisher={APS}
}

@article{haah2024learning,
  title={Learning quantum Hamiltonians from high-temperature Gibbs states and real-time evolutions},
  author={Haah, Jeongwan and Kothari, Robin and Tang, Ewin},
  journal={Nature Phys.},
  volume={20},
  number={arXiv: 2108.04842},
  pages={1027--1031},
  year={2024}
}

@article{gentile2021learning,
  title={Learning models of quantum systems from experiments},
  author={Gentile, Antonio A and Flynn, Brian and Knauer, Sebastian and Wiebe, Nathan and Paesani, Stefano and Granade, Christopher E and Rarity, John G and Santagati, Raffaele and Laing, Anthony},
  journal={Nature Physics},
  volume={17},
  number={7},
  pages={837--843},
  year={2021},
  publisher={Nature Publishing Group UK London}
}

@article{flynn2022quantum,
  title={Quantum model learning agent: characterisation of quantum systems through machine learning},
  author={Flynn, Brian and Gentile, Antonio A and Wiebe, Nathan and Santagati, Raffaele and Laing, Anthony},
  journal={New Journal of Physics},
  volume={24},
  number={5},
  pages={053034},
  year={2022},
  publisher={IOP Publishing}
}

@article{odake2024higher,
  title={Higher-order quantum transformations of Hamiltonian dynamics},
  author={Odake, Tatsuki and Kristj{\'a}nsson, Hl{\'e}r and Soeda, Akihito and Murao, Mio},
  journal={Physical Review Research},
  volume={6},
  number={1},
  pages={L012063},
  year={2024},
  publisher={APS}
}

@article{caro2024learning,
  title={Learning quantum processes and Hamiltonians via the Pauli transfer matrix},
  author={Caro, Matthias C},
  journal={ACM Transactions on Quantum Computing},
  volume={5},
  number={2},
  pages={1--53},
  year={2024},
  publisher={ACM New York, NY}
}

@inproceedings{haah2022optimal,
  title={Optimal learning of quantum Hamiltonians from high-temperature Gibbs states},
  author={Haah, Jeongwan and Kothari, Robin and Tang, Ewin},
  booktitle={2022 IEEE 63rd Annual Symposium on Foundations of Computer Science (FOCS)},
  pages={135--146},
  year={2022},
  organization={IEEE}
}

@article{zubida2021optimal,
  title={Optimal short-time measurements for Hamiltonian learning},
  author={Zubida, Assaf and Yitzhaki, Elad and Lindner, Netanel H and Bairey, Eyal},
  journal={arXiv preprint arXiv:2108.08824},
  year={2021}
}

@article{bairey2019learning,
  title={Learning a local Hamiltonian from local measurements},
  author={Bairey, Eyal and Arad, Itai and Lindner, Netanel H},
  journal={Physical review letters},
  volume={122},
  number={2},
  pages={020504},
  year={2019},
  publisher={APS}
}

@article{da2011practical,
  title={Practical characterization of quantum devices without tomography},
  author={da Silva, Marcus P and Landon-Cardinal, Olivier and Poulin, David},
  journal={Physical Review Letters},
  volume={107},
  number={21},
  pages={210404},
  year={2011},
  publisher={APS}
}

@article{shabani2011estimation,
  title={Estimation of many-body quantum Hamiltonians via compressive sensing},
  author={Shabani, Alireza and Mohseni, Masoud and Lloyd, Seth and Kosut, Robert L and Rabitz, Herschel},
  journal={Physical Review A—Atomic, Molecular, and Optical Physics},
  volume={84},
  number={1},
  pages={012107},
  year={2011},
  publisher={APS}
}

@article{giovannetti2004quantum,
  title={Quantum-enhanced measurements: beating the standard quantum limit},
  author={Giovannetti, Vittorio and Lloyd, Seth and Maccone, Lorenzo},
  journal={Science},
  volume={306},
  number={5700},
  pages={1330--1336},
  year={2004},
  publisher={American Association for the Advancement of Science}
}

@article{de2005quantum,
  title={Quantum methods for clock synchronization: Beating the standard quantum limit without entanglement},
  author={De Burgh, Mark and Bartlett, Stephen D},
  journal={Physical Review A—Atomic, Molecular, and Optical Physics},
  volume={72},
  number={4},
  pages={042301},
  year={2005},
  publisher={APS}
}

@article{ramsey1950molecular,
  title={A molecular beam resonance method with separated oscillating fields},
  author={Ramsey, Norman F},
  journal={Physical Review},
  volume={78},
  number={6},
  pages={695},
  year={1950},
  publisher={APS}
}

@article{lee2002quantum,
  title={A quantum Rosetta stone for interferometry},
  author={Lee, Hwang and Kok, Pieter and Dowling, Jonathan P},
  journal={Journal of Modern Optics},
  volume={49},
  number={14-15},
  pages={2325--2338},
  year={2002},
  publisher={Taylor \& Francis}
}

@article{anshu2021sample,
  title={Sample-efficient learning of interacting quantum systems},
  author={Anshu, Anurag and Arunachalam, Srinivasan and Kuwahara, Tomotaka and Soleimanifar, Mehdi},
  journal={Nature Physics},
  volume={17},
  number={8},
  pages={931--935},
  year={2021},
  publisher={Nature Publishing Group UK London}
}
